\documentclass[twocolumn,10pt]{IEEEtran}
\usepackage{braket}

\input{def.tex}
\usepackage{dsfont}
\usepackage{minibox}
\DeclareSymbolFont{matha}{OML}{txmi}{m}{it}
\DeclareMathSymbol{\varv}{\mathord}{matha}{118}
\usepackage{eurosym}
\usepackage{hhline,physics}

\usepackage{multicol}
\usepackage{algorithm}
\usepackage{graphicx}
\usepackage{textcomp}
\usepackage{caption,pifont}
\usepackage[multiple]{footmisc}

\usepackage{algorithmic}

\newcommand{\trans}{^{\mathsf{T}}}
\newcommand{\herm}{^{\H}}
\makeatletter

\IEEEoverridecommandlockouts
\begin{document}
	\title{STAR-RIS-Assisted Communication Radar Coexistence: Analysis and Optimization} 
		\author{Anastasios Papazafeiropoulos,  Pandelis Kourtessis, Symeon Chatzinotas \thanks{A. Papazafeiropoulos is with the Communications and Intelligent Systems Research Group, University of Hertfordshire, Hatfield AL10 9AB, U. K., and with SnT at the University of Luxembourg, Luxembourg.   P. Kourtessis is with the Communications and Intelligent Systems Research Group, University of Hertfordshire, Hatfield AL10 9AB, U. K.   S. Chatzinotas is with the SnT at the University of Luxembourg, Luxembourg. A. Papazafeiropoulos was supported  by the University of Hertfordshire's 5-year Vice Chancellor's Research Fellowship.  It has also received funding from the European Research Council (ERC) under the European Union’s Horizon 2020 research and innovation programme (Grant agreement No. 819819).
			S. Chatzinotas   was supported by the National Research Fund, Luxembourg, under the project RISOTTI. E-mails: tapapazaf@gmail.com, p.kourtessis@herts.ac.uk, symeon.chatzinotas@uni.lu.}}
	\maketitle\vspace{-1.7cm}
	\begin{abstract}
		Integrated sensing and communication (ISAC) is expected to play a prominent role among emerging technologies in future wireless communications. In particular,  a communication radar coexistence system is degraded significantly by mutual interference. In this work, given the advantages of promising reconfigurable intelligent surface (RIS), we propose  a simultaneously transmitting and reflecting RIS (STAR-RIS)-assisted radar coexistence system where a STAR-RIS is introduced to improve  the communication performance while suppressing the  mutual interference and providing full space coverage. Based on the realistic conditions of   correlated fading, and the presence of multiple user equipments (UEs) at both sides of the RIS, we derive the achievable rates at the radar and the communication receiver side in closed forms in terms of statistical channel state information (CSI). Next, we perform alternating optimization (AO) for optimizing the STAR-RIS and the radar beamforming. Regarding the former, we optimize the amplitudes and phase shifts of the STAR-RIS  through a projected gradient ascent algorithm (PGAM) simultaneously with respect to the amplitudes and phase shifts of the surface for both energy splitting (ES) and mode switching (MS)  operation protocols. The proposed optimization saves enough overhead since it can be performed every several coherence intervals. This property is  particularly beneficial compared to  reflecting-only RIS because a STAR-RIS includes the double number of variables, which   require increased overhead. Finally, simulation results illustrate how the proposed architecture outperforms the conventional RIS counterpart, and show how the various parameters affect the performance. \textcolor{black}{Moreover, a benchmark full instantaneous CSI  (I-CSI) based design is provided and shown to result in higher sum-rate but also in large overhead associated with complexity.}
	\end{abstract}
	
	\section{Introduction}
	The exponential growth of mobile devices and the gradual increase in the data rate demand are supported by the rapid development of future wireless communication networks \cite{Tariq2020}. The required improvement in wireless transmission capacity can be achieved in two ways. The first key way concerns the move to a higher spectrum such as millimeter wave or even terahertz transmission \cite{Wang2018a,Elayan2018}. The second promising direction includes sharing of the frequency bandwidth among different systems. In parallel, intelligent wireless communications including artificial intelligence (AI) require massive sensing data for the training of the AI network. Note that the acquirement of such data can be achieved by radar systems. To cover these gaps, integrated sensing and communication (ISAC) has emerged as a promising solution that can improve wireless capacity and collect sensing data simultaneously \cite{Luong2021,Liu2022}. 
	
	Generally, there are two directions with respect to ISAC implementation. One is the dual-functional radar and communication (DFRC) system design \cite{Liu2018,Hassanien2019,Liu2020a}, where the communication and radar devices share the same hardware. The other includes a communication radar coexistence system \cite{Zheng2019a}, where the communication and radar devices are separated but share the same frequency bandwidth. In this work, we rely on the latter implementation with the relevant works focusing on the suppression of the mutual interference between the two systems by means of suitable communication beamforming and radar codesign. Notably, the mutual interference can be mitigated by the introduction of a disruptive simultaneously transmitting and reflecting reconfigurable intelligent surfaces (STAR-RIS) technology in three ways as will be proposed below, communication signal enhancement and mutual interference suppression while providing $ 360^{\circ} $ coverage. Note that it has been shown that a RIS not only can play an important role in spectrum sharing in cognitive systems \cite{Yuan2021} but a conventional RIS has been already employed in communication radar coexistence system \cite{He2022}.
	
	Generally, A RIS consists of a metasurface that includes a large number of passive elements which can be reconfigured independently by inducing suitable phases on the impinging waves. This technology has been introduced as a cost- and energy-efficient solution to improve the performance in wireless systems \cite{Wu2019,Basar2019}. An accompanied controller is responsible for the management or the phase shifts towards the coherent addition of the reflected signals that will result in boosting the desired signal at the receiver. 
	
	Although the main assumption in most existing works is that both the transmitter and the receiver are located on the same side of the surface with only reflection taking place, in practice, user equipments (UEs) can be located not only in front but also behind the surface. Fortunately, recent advancements in metamaterials have enabled STAR-RIS that provides full space coverage by adjusting both the amplitudes and phase shifts of the impinging waves \cite{Xu2021,Mu2021,Niu2021,Wu2021,Niu2022,Papazafeiropoulos2023,Papazafeiropoulos2023a}. Especially, in \cite{Mu2021}, authors suggested the operating protocols for modifying the transmission and reflection coefficients of the transmitting and reflected signals which are namely energy splitting (ES), mode switching (MS), and time switching (TS). \textcolor{black}{Note that the performance analysis of the ES/MS and TS protocols require different lines of analysis. For example, in \cite{Papazafeiropoulos2023}, the downlink achievable rate and its optimization of a STAR-RIS-assisted mMIMO system was obtained, while in \cite{Papazafeiropoulos2023a}, the max-min SINR of  a STAR-RIS-assisted mMIMO system with hardware impairments was studied for the ES and MS protocols.}
	
	To address the interference between communication and radar systems, sophisticated optimization techniques by means of joint spectrum allocation \cite{Wang2019} and beamforming design have been studied \cite{Liu2018a}. Moreover, a number of works have considered the introduction of RIS into ISAC \cite{Wang2021c,Sankar2021,Jiang2021,Wang2020a,He2022}. In \cite{Wang2021c}, a RIS-assisted MIMO DFRC system was considered to reduce the mutual reference by deploying the RIS close to the communication devices by performing alternating optimization (AO)-based method. In \cite{Sankar2021}, a RIS was adaptively partitioned in two parts for communication and localization respectively by means of a large localization algorithm and a RIS passive beamforming (PB) algorithm. In \cite{Jiang2021}, a RIS was employed for both communication and sensing by maximising the signal-to-interference-plus-noise ratio (SINR) constraint under a communication constraint while developing an AO approach for the corresponding problem. In \cite{Wang2020a}, authors focused on a coexistence system and considered the spectrum sharing problem between MIMO and multi-user RIS-assisted communication systems,  while in \cite{He2022}, a double RIS considered the performance improvements with a single receiver and a single-antenna transmitter.
	
	In general, there are two methods for the optimization of the phase shifts. Specifically, the first approach relies on the instantaneous channel state information  CSI  (I-CSI) \cite{Wu2019,Pan2020}, and the second approach relies on statistical CSI \cite{Zhao2020,Kammoun2020,Papazafeiropoulos2021,VanChien2021,Papazafeiropoulos2021a}. The former approach assumes that the phase shifts are optimized at every coherence interval because the corresponding analytical results depend on small-scale fading. On the other hand, the latter approach includes expressions in terms of large-scale fading statistics that vary at every several coherence intervals. Thus, the corresponding optimization can be performed less frequently, which reduces the signal overhead and induces lower complexity. Especially, when the number of RIS elements becomes very large or in high-speed scenarios with fast time-varying channels, the first method becomes prohibitive. In addition, in the case of the first approach, the continuous overloading of the smart controller is not energy efficient. \textcolor{black}{Note that the second approach
	is basically based on  the  two-timescale transmission protocol approach  \cite{Zhao2020}. According to this protocol, the precoding is designed in terms of I-CSI, while the RIS   phase shifts are optimised by using statistical CSI. Other examples of works, which are based on the two-timescale protocol are the study of the impact of hardware impairments on the sum rate and the minimum rate in  \cite{Papazafeiropoulos2021} and \cite{Kammoun2020}, respectively.}
	
	In this work, we consider a typical scenario, where a communication multi-user multi-input single-output (MISO) architecture shares the same frequency with a radar. Unlike previous works that have considered conventional reflecting-only RIS, we propose a STAR-RIS-assisted radar coexistence system with correlated fading, where the surface aims to suppress the interference from the transmitter to the radar and cancel the interference from the radar to the UEs while using the two-timescale protocol \cite{Zhao2020,Papazafeiropoulos2021}.\footnote{ In \cite{Bjoernson2020}, it was shown that the RIS correlation should be taken into account because it is unavoidable in practice.} \textcolor{black}{In parallel to this work, we noticed \cite{Wang2022,Wang2023},  which considered the coexistence of  STAR-RIS with a radar. \textcolor{black}{However, our work differentiates in a number of ways. Specifically,our work is the only one that has studied the coexistence of the STAR-RIS with a radar with statistical CSI by using the two-timescale protocol. In \cite{Wang2022,Wang2023}, only I-CSI was considered. Hence, in this work, the optimization can take place at every several coherence intervals, not at each coherence interval. Also, the system models and the performance metric are different.  Contrary to  \cite{Wang2022,Wang2023}, we have assumed realistic correlated Rayleigh fading. Spatial correlation makes  the formulation more complicated since it concerns manipulations with matrices instead of scalars (path-losses) used in many other works. In the optimization, the correlation matrices induce special manipulations to compute the gradient ascent. Moreover, we have shown its importance to perform phase shifts optimization. Also, the introduction of spatial correlation increases the computational complexity of the optimization since it includes matrix-wise operations instead of scalar operations.\footnote{\textcolor{black}{ In \cite{Wang2022,Wang2023}, the whole space is divided into two half-spaces, namely the sensing space and the communication space. In our case,  no space division takes place. Both areas of the STAR-RIS serve all communication and sensing users. In \cite{Wang2022,Wang2023}, one target is assumed. We assume that each user can be potential a target and a communication user. In \cite{Wang2022,Wang2023}, low-cost sensors on STAR- RIS were installed.  In our work, we have have not applied such installation.    In \cite{Wang2022,Wang2023}, the Cramer-Rao bound (CRB) the two-dimensional (2D) direction-of-arrivals (DOAs) estimation of the sensing target was derived, while we have derived the sum SE and performed its optimization.}} In addition, note that many previous works have assumed that the channel between the BS and the RIS is deterministic expressing a line-of-sight (LoS) component, e.g,. [26], while the analysis here is more general since we assume that all links are correlated Rayleigh fading distributed. We have provided an analytical framework to present in an elegant unified way the analysis of both types of UEs located in front of and the behind the surface. Furthermore, we have considered multiples UEs at each side of the surface, while most works on STAR-RIS have assumed only one UE at each side. Moreover, we have managed not only to provide the gradient ascent in closed form but we have also applied it simultaneously on the amplitudes and the phase shifts, while other works optimize only the phase shifts or optimize both the amplitudes and the phase shifts in an  alternating optimization manner.	}	Based on AO, we aim to maximise the communication performance in terms of the passive beamforming at the STAR-RIS and the radar beamforming, while having a constraint regarding the radar transmit power that guarantees the radar detection performance.}  The main contribution are summarised as follows:
	
	\begin{itemize}
		\item To suppress the mutual interference while providing $ 360^{\circ} $ coverage, we propose a novel STAR-RIS-assisted  radar coexistence system with correlated fading at both the BS and the STAR-RIS, where multiple UEs are located at each side of the surface based on the  two-timescale protocol.
		
		\item We derive the achievable spectral efficiencies (SEs) with correlated fading at the radar and the desired UE located at the transmit or the reflection side of the surface in closed forms that depend only on large-scale statistics.
		\item Based on statistical CSI, where the analytical results depend only on large-scale statistics that change at every several coherence intervals, we formulate the optimization problem to maximise the achievable SE, while avoiding interference from the radar. Also, we perform simultaneous optimization of the amplitudes and the phase shifts of the STAR-RIS, which significantly reduces the overhead. This property is quite important for STAR-RIS applications that have twice the number of optimization variables compared to conventional RIS.
		\item Simulation results are provided to validate the analytical results and shed light on the impact of various parameters. Also, the superiority of the proposed architecture against the conventional counterpart is shown. In addition, we consider as a performance benchmark a full  I-CSI based RIS design, and we depict the superiority of the proposed framework based on S-CSI.
	\end{itemize}
	
	\textit{Paper Outline}: The remainder of this paper is organized as follows. Section~\ref{System} presents the system model of a STAR-RIS-assisted  radar coexistence system with correlated Rayleigh fading. Section~\ref{PerformanceAnalysis} presents the downlink data transmission. Section \ref{SumSEMaximizationDesign} provides the simultaneous optimization of the  amplitudes and  phase-shifts configuration concerning both  the PBs for the  transmission and reflection regions as well as the optimization of the radar beamforming design. The numerical results are placed in Section~\ref{Numerical}, and Section~\ref{Conclusion} concludes the paper.
	
	\textit{Notation}: Vectors and matrices are denoted by boldface lower and upper case symbols, respectively. The notations $(\cdot)^\T$, $(\cdot)^\H$, and $\tr\!\left( {\cdot} \right)$ describe the transpose, Hermitian transpose, and trace operators, respectively. Moreover, the notation  $\EE\{\cdot\}$ expresses   the expectation operator. The notation  $\diag\left(\bA\right) $ describes a vector with elements equal to the  diagonal elements of $ \bA $, the notation  $\diag\left(\bx\right) $ describes a diagonal  matrix whose elements are $ \bx $, while  $\bb \sim \cC\cN{(\b0,\mathbf{\Sigma})}$ describes a circularly symmetric complex Gaussian vector with zero mean and a  covariance matrix $\mathbf{\Sigma}$. 
	
	\section{System Model}\label{System}
	In this section, we present the model of the STAR-RIS-assisted communication radar coexistence system, as shown in Fig. \ref{Fig1}. Specifically, we consider a BS with an $ M $-element uniform linear array (ULA) that serves simultaneously $ K $ single-antenna UEs, which are located on both sides of a STAR-RIS with $ N  $ elements,  belonging to the set  $\mathcal{N}=\{1, \ldots, N \}$. We assume that 	$\mathcal{K}_{t}=\{1,\ldots,K_{t} \} $ UEs are distributed  in the transmission region $ (t) $ and $ \mathcal{K}_{r}=\{1,\ldots,K_{r} \} $ UEs are  distributed in the reflection region $ (r) $ of the STAR-RIS, respectively, i.e., $ K_{t}+K_{r}=K $. To define the RIS operation mode for each of the $ K $ UEs, we denote the set $ \mathcal{W} = \{w_{1}, w_{2}, ..., w_{K}\}  $. If the  $ k $th UE is found behind the STAR-RIS (i.e., $ k\in   \mathcal{K}_{t}$), then $ w_{k} = t $. On the contrary,  if the  $ k $th UE is found in front of  the STAR-RIS (i.e., $ k\in   \mathcal{K}_{r}$), we have  $ w_{k} = r $. The radar is deployed with $ Q $ transmit and receive antennas. The communication nodes and the radar operate in the same frequency band. Moreover, we assume direct links between the BS and UEs.  
	\begin{figure}[!h]
		\begin{center}
			\includegraphics[width=0.7\linewidth]{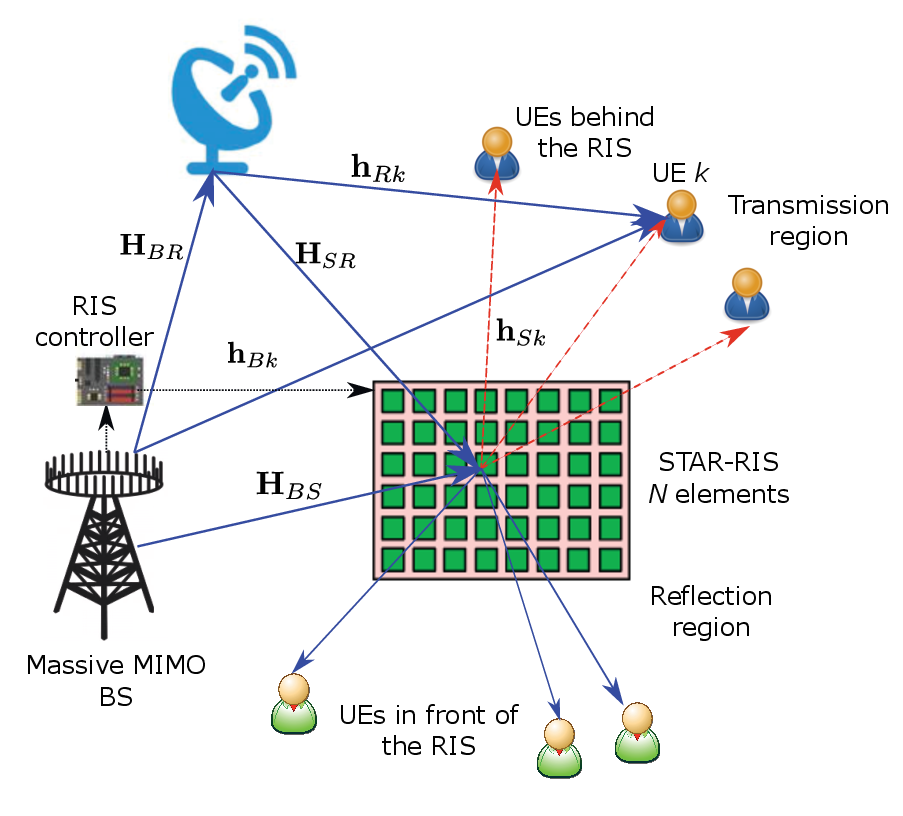}
			\caption{{ A STAR-RIS-assisted communication radar coexistence system. }}
			\label{Fig1}
		\end{center}
	\end{figure} 
	\subsection{STAR-RIS Model}
	The STAR-RIS is implemented by a uniform planar array (UPA) that consists of $ N_{\mathrm{h}} $ horizontally 	and $ N_{\mathrm{v}} $ vertically  passive elements, i.e.,  $  N_{\mathrm{h}} \times N_{\mathrm{v}}=N $. The operation of the surface is based on two independent coefficients, which can configure  the transmited ($ t $) and reflected ($ r $) signals. Especially, \textcolor{black}{given the impinging signal $ s_{n} $,} the transmitted 	and reflected signal by the $ n $th STAR-RIS element are denoted as  $ t_{n} =( {\beta_{n}^{t}}e^{j \phi_{n}^{t}})s_{n}$ and $ r_{n}=( {\beta_{n}^{r}}e^{j \phi_{n}^{r}})s_{n} $, respectively. Note that  $ {\beta_{n}^{w_{k}}}\in [0,1] $ and $ \phi_{n}^{w_{k}} \in [0,2\pi)$ correspond to the amplitude and phase parameters, while the UE can be found in any of the two regions that are geographically independent. \textcolor{black}{The impact of the phase estimation error, which is known as phase noise, studied in \cite{Papazafeiropoulos2021} for conventional RIS and in \cite{Papazafeiropoulos2022a} STAR-RIS, is left for future work}. The amplitudes are correlated according to the law of energy conservation as 
	\begin{align}
		(\beta_{n}^{t})^{2}+(\beta_{n}^{r})^{2}=1,  \forall n \in \mathcal{N}
	\end{align}
	but the phase shifts $ \phi_{n}^{t} $ and $ \phi_{n}^{r} $ can be chosen independently.\footnote{\textcolor{black}{Recently, a practical coupled phase-shift model was shown, while the assumption of  independent phase shifts is impractical, e.g., see \cite{Liu2022a,Wu2021}. However, this assumption regarding independence still allows revealing fundamental properties of the coexistence of radar and STAR-RIS, while the consideration of the practical phase shift model in \cite{Liu2022a,Wu2021} is an interesting idea for extension of the current work, i.e., to study the coexistence of radar and STAR-RIS by accounting for this dependence.}} For the sake of exposition, we denote $ \theta_{i}^{w_{k}}=e^{j\phi_{i}^{w_{k}}} $.
	\subsubsection{Operation protocols} The following analysis relies on the ES/MS protocols  \cite{Mu2021}. \textcolor{black}{The TS protocol, requiring a separate analysis, could be the topic of future work. } The main points of  ES and MS protocols are outlined below.
	
	\textit{\textbf{ES protocol:}} According to this protocol, all STAR-RIS elements serve simultaneously all UEs  located in both   $ t $ and $ r $ regions. The PB for the $ k
	$th  UE is expressed as $ \bPhi_{w_{k}}^{\mathrm{ES}}=\diag( {\beta_{1}^{w_{k}}}\theta_{1}^{w_{k}}, \ldots,  {\beta_{N}^{w_{k}}}\theta_{N}^{w_{k}}) \in \mathbb{C}^{N\times N}$, where $ \beta_{n}^{w_{k}} \ge 0 $, $ 		(\beta_{n}^{t})^{2}+(\beta_{n}^{r})^{2}=1 $, and $ |\theta_{n}^{w_{k}}|=1, \forall n \in \mathcal{N} $.
	
	\textit{\textbf{MS protocol:}} This protocols assumes that the surface is partitioned  into two groups of $ N_{t} $ and $ N_{r} $ elements that serve UEs in the  $  t$ and $ r $ regions, respectively, i.e.,  $ N_{t}+N_{r}=N $. The PB is given by $ \bPhi_{w_{k}}^{\mathrm{MS}}=\diag( {\beta_{1}^{w_{k}}} \theta_{1}^{w_{k}}, \ldots,  {\beta_{N}^{w_k}}\theta_{N}^{w_{k}}) \in \mathbb{C}^{N\times N}$, where $ \beta_{n}^{w_{k}}\in \{0,1\}$, $ 	(\beta_{n}^{t})^{2}+(\beta_{n}^{r})^{2}=1 $, and $ |\theta_{i}^{w_{k}}|=1, \forall n \in \mathcal{N} $. Given that the amplitude coefficients for transmission and reflection are restricted to binary values, the  MS protocol is inferior to the ES counterpart. However, MS has lower computational complexity regarding the  PB design.\footnote{\textcolor{black}{In this case, the optimal value regarding the partitioning of the surface could be the focus of future research.}}
	
	\subsection{Radar Model}
	During, a radar detection epoch, we assume $Z $	detection directions denoted by $ \{\bar{\theta}_{z}\}, \forall z  \in \mathcal{Z}=\{1, \dotsb, Z\}$. The pulse  repetition interval (PRI), which expresses the length of detection time in each direction $ \bar{\theta}_{z} $, is given by $ L $. Hence, the duration  of one detection epoch is $ ZL $. During each PRI, the transmission of the probing pulse takes place at the first time index, which is $ l=1 $, while the echo signal is received at the time index $ \bar{l} $. The probing pulse of the radar is expressed as
	\begin{align}
		\bar{\bx}[l]=\left\{
		\begin{array}{ll}
			\bar{\bu}_{z}\bar{s}_z &l=(z-1)L+1,\forall z \\
			\b0 & l\ne (z-1)L+1,\forall z  \\
		\end{array} 
		\right.
	\end{align}
	where $ \bar{s}_z $ with $ \EE\{\bar{s}_z^{*}\bar{s}_z  \} =1$ and $ \bar{\bu}_{z} \in \mathbb{C}^{Q \times 1}$ correspond to the radar signal and the radar transmit beamforming vector for direction $ \bar{\theta}_{z} $, respectively.

	\subsection{Channel Model}\label{ChannelModel} 
	We assume  narrowband quasi-static block fading channels with each block having a duration of $\tau_{\mathrm{c}}$ channel uses. 
	
	Let  $ \bH_{BS} \in \mathbb{C}^{M \times N}$, $ \bH_{BR} \in \mathbb{C}^{M \times Q}$, $ \bh_{Rk} \in \mathbb{C}^{Q \times 1}$, $ \bH_{SR} \in \mathbb{C}^{N \times Q}$,  $ \bh_{Sk} \in \mathbb{C}^{N \times 1}$, and $ \bh_{Bk} \in \mathbb{C}^{M \times 1}$ be the channels between the BS and the STAR-RIS, the BS and the radar, the radar and UE $ k $, the STAR-RIS and the radar,  the STAR-RIS and UE $ k $, and the BS and UE $ k $. All these links underlie correlated Rayleigh fading conditions since spatial correlation is unavoidable in practice. The study of Rician fading, which includes a line-of-sight (LoS) component, is the topic of future research. Specifically, we have
	\begin{align}
		\bH_{BS}&\!=\!\sqrt{\!\tilde{ \beta}_{BS}}\bR_{\mathrm{BS}}^{1/2}\bZ_{BS}\bR_{\mathrm{RIS}}^{1/2},\bH_{BR}\!=\!\sqrt{\!\tilde{ \beta}_{BR}}\bR_{\mathrm{BS}}^{1/2}\bZ_{BR}\bR_{\mathrm{R}}^{1/2},\label{eq2}\\
		\bh_{Rk}&\!=\!\sqrt{\tilde{ \beta}_{Rk}}\bR_{\mathrm{R}}^{1/2}\bz_{Rk},~~~~~~~~\!\bH_{SR}\!=\!\sqrt{\tilde{ \beta}_{SR}}\bR_{\mathrm{RIS}}^{1/2}\bZ_{SR}\bR_{\mathrm{R}}^{1/2},\label{eq3}\\
		\bh_{Sk}&\!=\!\sqrt{\tilde{ \beta}_{Sk}}\bR_{\mathrm{RIS}}^{1/2}\bz_{Sk},~~~~~~\!~~~\bh_{Bk}\!=\!\sqrt{\tilde{ \beta}_{Bk}}\bR_{\mathrm{BS}}^{1/2}\bz_{Bk},\label{eq4}
	\end{align}
	where the correlation matrices $ \bR_{\mathrm{BS}} \in \mathbb{C}^{M \times M} $, $ \bR_{\mathrm{RIS}} \in \mathbb{C}^{N \times N} $, and  $ \bR_{\mathrm{R}} \in \mathbb{C}^{Q \times Q}$ are Hermitian-symmetric positive semi-definite, and correspond to the BS, the RIS, and the radar, respectively. They can be obtained by existing estimation methods \cite{Neumann2018,Upadhya2018}. Hence, we assume that they are known by the network. Both  $  \bR_{\mathrm{BS}} $ and $  \bR_{\mathrm{R}} $ can be modeled  as in \cite{Hoydis2013}, while $ \bR_{\mathrm{RIS}} $ can be  modeled as in \cite{Bjoernson2020}. \textcolor{black}{Specifically, we assume that $  \bR_{\mathrm{BS}}=[\bA~\b0_{M\times M-P}]$, where $ \bA=[\ba(\varphi_{1}) \cdots \ba(\varphi_{P})] \in \mathbb{C}^{N\times P}$ is composed of the steering vector $ \ba(\varphi)\in \mathbb{C}^{N} $ defined as
	\begin{align}
		\ba(\varphi)=\frac{1}{\sqrt{P}}[1, e^{-i2\pi\omega \sin(\varphi), \ldots,-i2\pi\omega (N-1)\sin(\varphi)}]^{\T}, 
	\end{align}
where $ \omega =0.3 $ is the antenna spacing in multiples of the wavelength and $ \varphi_{p}=-\pi/2+(p-1)\pi/P $, $ p= 1, \ldots, P$ are the uniformly distributed angles of transmission with $ P=M/2 $ \cite{Hoydis2013}. Similarly, we describe $  \bR_{\mathrm{R}} $.  Also, the $ (l,m) $th entry of $ \bR_{\mathrm{RIS}} $ is described as
	\begin{align}
		[\bR_{\mathrm{RIS}}]_{l,m}=\mathrm{sinc}\Big(\frac{2\|\bu_{l}-\bu_{m}\|}{\lambda}\Big), ~~\forall \{l,m\}\in \mathcal{N},
	\end{align}
	where $ \|\bu_{l}-\bu_{m}\| $ expresses the distance between the $ l $th and $ m $th  STAR-RIS elements, and $ \lambda $ is the wavelength \cite{Bjoernson2020}.}\footnote{\textcolor{black}{Another way of practical calculation of the covariance matrices follows. In particular,  the covariance matrices depend on the distances between the RIS elements and between the BS antennas, respectively. Moreover,  they depend on  the angles. Generally, the distances are known from the construction of the  RIS and the BS, while  the angles can be calculated when the 	locations are given.  Thus, the covariance matrices can be considered to be known.}}
	
	 Also, $ \tilde{ \beta}_{BS} $,  $ \tilde{ \beta}_{BR} $, $ \tilde{ \beta}_{Rk} $, $ \tilde{ \beta}_{SR} $,  $ \tilde{ \beta}_{Sk} $,  $ \tilde{ \beta}_{Bk} $, and $ \mathrm{vec}(\bZ_{BS})\sim \mathcal{CN}\left(\b0,\Id_{MN}\right) $, $ \mathrm{vec}(\bZ_{BR})\sim \mathcal{CN}\left(\b0,\Id_{MQ}\right) $,  $ \bar{\bz}_{Rk} \sim \mathcal{CN}\left(\b0,\Id_{Q}\right) $, $ \mathrm{vec}(\bZ_{SR})\sim \mathcal{CN}\left(\b0,\Id_{NQ}\right) $, $ \bar{\bz}_{Sk} \sim \mathcal{CN}\left(\b0,\Id_{N}\right) $, $ \bar{\bz}_{Bk} \sim \mathcal{CN}\left(\b0,\Id_{M}\right) $ are the corresponding path losses and fast fading vectors, respectively.

	\section{Data Transmission}\label{PerformanceAnalysis}
	
	Regarding the radar, the echo from the target in the direction $ \bar{\theta}_{z} $ is written as
	\begin{align}
		\bar{\by}_{z}[l]=\al_{z}\ba(\bar{\theta}_{z})\ba^{\T}(\bar{\theta}_{z})\bar{\bx}[l-\bar{l}],\label{echo}
	\end{align}
	where  $ \ba(\bar{\theta}) =[1, e^{j\frac{2\pi \Delta}{\lambda}\sin(\theta)}, \ldots, e^{j\frac{2\pi \Delta}{\lambda}(Q-1)\sin(\theta)}]$ with  $ \Delta $ being the antenna spacing. \textcolor{black}{Note that  we have assumed that the STAR-RIS is placed near the communication transmitter and receiver and far away from the radar. The reflecting signals are only being considered when a RIS is placed near the transmitter or receiver. }
	
	The interference from the BS to the radar consists of two links, which correspond to the BS-radar channel and the BS-STAR-RIS-radar channel.  Mathematically, this interference is written as
	\begin{align}
		\tilde{\by}_{z}[l]=(\bH_{BR}^{\H}+\bH_{SR}^{\H}\bPhi_{w_{k}}^{\H}\bH_{BS}^{\H})\bs[l],\label{interference}
	\end{align}
	where $\bs[l]=\sqrt{\lambda} \sum_{i=1}^{K}\sqrt{p_{i}}\bff_{i}x_{i}[l]$ expresses the transmit signal vector  by the BS,  
	and $ \bar{\lambda} $ is a constant which is found such that $ \EE[\mathbf{s}[l]^{\H}\mathbf{s}[l]]=\rho $ with $ \rho $ being the total average power budget. Note that $\bff_{i} \in \bbC^{M \times 1}$ is the linear precoding vector and $ x_{i}[l] $ is     the corresponding data symbol with $ \EE\{|x_{i}[l]|^{2}\}=1 $. 
	Hence, due to $  \EE[\mathbf{s}[l]^{\H}\mathbf{s}[l]]=\rho$, $ \bar{\lambda} $ is given as  $ \bar{\lambda}=\frac{K}{\EE\{\tr\textcolor{black}{(\bF\bF^{\H})}\}} $, where $ \bF=[\bff_{1}, \ldots, \bff_{K}] \in\mathbb{C}^{M \times K}$.

	Thus, the received signal by the radar is given based on \eqref{echo} and \eqref{interference} by
	\begin{align}
		\by_{z}[l]&=\al_{z}\ba(\bar{\theta}_{z})\ba^{\T}(\bar{\theta}_{z})\bar{\bx}[l-\bar{l}]\nn\\
		&+(\bH_{BR}^{\H}+\bH_{SR}^{\H}\bPhi_{w_{k}}^{\H}\bH_{BS}^{\H})\bs[l]+\bn[l],\label{received1}
	\end{align}
	where $ \bn[l] \sim \mathcal{CN}(\b0,\sigma_{r}^{2}\Id_{Q})$ is the complex circular Gaussian
	noise vector that consists of additive white Gaussian noise and clutter. Therein, we apply  a receive beamforming vector $ \bw_{z} $ to detect the echo from direction $ \bar{\theta}_{z} $.
	
	The SINR of \eqref{received1} can be obtained by treating the interference from the BS as independent Gaussian noise, which is a worst-case assumption when computing the mutual information \cite{Hassibi2003}. Note that the effective channel is $\bw_{z} \al_{z}\ba(\bar{\theta}_{z})\ba^{\T}(\bar{\theta}_{z}\textcolor{black}{)}  \bar{\bu}_{z} $. Hence, the SINR is written as in \eqref{SINR1} at the top of the next page.

\begin{figure*}
\begin{align}
		\gamma_{z}=\frac{|\bw_{z}^{\H} \al_{z}\ba(\bar{\theta}_{z})\ba^{\T}(\bar{\theta}_{z} ) \bar{\bu}_{z}|^{2}}{\frac{\bar{\lambda} \rho}{K}\EE \{|\bw_{z}^{\H}(\bH_{BR}^{\H}+\bH_{SR}^{\H}\bPhi_{w_{k}}^{\H}\bH_{BS}^{\H}) \sum_{i=1}^{K}\bff_{i}|^{2}\}+\sigma^{2}_{r}\EE\{\|\bw_{z}\|^{2}\}}\label{SINR1}.
	\end{align}
	\hrulefill
\end{figure*}

	In the case of the communication between the BS and UE $ k $ via the STAR-RIS, two links exist. Specifically, we have the direct channel  between the BS and UE $ k $, and the BS-STAR-RIS-UE $ k $ channel. In parallel, the signal from the radar  appears as interference to UE $ k $ via two links, which are the radar-UE $ k $ link and the radar-STAR-RIS-UE $ k $ link. As a result, the received signal at UE $ k $ at time index $ l $ can be written as
	\begin{align}
		y\textcolor{black}{_{k}}[l]&=(\bh_{Bk}^{\H}+\bh_{Sk}^{\H}\bPhi_{w_{k}}^{\H}	\bH_{BS}^{\H})\bs[l]\nn\\
		&+(\bh_{Rk}^{\H}+\bh_{Sk}^{\H}\bPhi_{w_{k}}^{\H}\bH_{SR})\bar{\bx}[l]+n_{c}[l],
	\end{align}
	where the first part corresponds to the communication signal, the second part is the interference from the radar, and $ n_{c}[l]\sim \mathcal{CN}(0, \sigma_{c}^{2}) $ is the complex white Gaussian noise at UE $ k $ with zero mean and variance $ \sigma_{c}^{2} $.
	
	The overall energy of the received signal during one radar detection epoch is obtained as
	\begin{align}
		\EE\{\sum_{l=1}^{ZL} y\textcolor{black}{_{k}}[l]^{\H}y\textcolor{black}{_{k}}[l]\}&=\frac{ZL \bar{\lambda} \rho}{K}\sum_{i=1}^{K}|\bh_{BSk}^{\H}\bff_{i}|^{2}\nn\\
		&+\sum_{z=1}^{Z}|\bh_{RSk}^{\H}\bar{\bu}_{z}|^{2}+ZL\sigma_{c}^{2}.
	\end{align}

	Based on the use-and-then forget (UaTF) technique \cite{massivemimobook}, the downlink achievable SINR over one radar detection epoch is given by
	\eqref{SINR2}
	\begin{figure*}
			\begin{align}
		\gamma_{k}=\frac{|\EE\{\bh_{BSk}^{\H}\bff_{k}\}|^{2}}{\EE\{|\bh_{BSk}^{\H}\bff_{k}\!-\!		\EE\{\bh_{BSk}^{\H}\bff_{k}\}|^{2}\}\!+	\!	\sum_{i\ne k}\EE\{|\bh_{BSk}^{\H}\bff_{i}|^{2}\}\!+\!\frac{K}{\bar{\lambda} \rho ZL}\sum_{z=1}^{Z}\EE\{|\bh_{RSk}^{\H}\bar{\bu}_{z}|^{2}\}\!+\!\frac{\sigma_{c}^{2}K}{\bar{\lambda} \rho}}\label{SINR2}.
	\end{align}
	\hrulefill
\end{figure*}
	where we have denoted $ \bh_{BSk}=\bh_{Bk}+	\bH_{BS}\bPhi_{w_{k}}\bh_{Sk}\sim\mathcal{CN}(\b0, \bR_{BSk})$. Note that $ \bR_{BSk}=\tilde{ \beta}_{Bk}\bR_{\mathrm{BS}}+\tilde{ \beta}_{BSk}\tr(\bR_{\mathrm{RIS}}\bPhi_{w_{k}}\bR_{\mathrm{RIS}}\bPhi_{w_{k}}^{\H})\bR_{\mathrm{BS}} $ with $\tilde{ \beta}_{BSk}= \tilde{ \beta}_{BS}\tilde{ \beta}_{Sk} $. Also, we denote  $ \bh_{RSk}=\bh_{Rk}+\bH_{SR}^{\H}\bPhi_{w_{k}}\bh_{Sk}\sim\mathcal{CN}(\b0, \bR_{RSk})$, where $ \bR_{RSk}=\tilde{ \beta}_{Rk}\bR_{\mathrm{R}}+\tilde{ \beta}_{RSk}\tr(\bR_{\mathrm{RIS}}\bPhi_{w_{k}}\bR_{\mathrm{RIS}}\bPhi_{w_{k}}^{\H}) \bR_{\mathrm{R}}$ with $\tilde{ \beta}_{RSK}=\tilde{ \beta}_{SR}\tilde{ \beta}_{Sk} $
	The final expressions of \eqref{SINR1} and \eqref{SINR2}  depend on  the choice of the decoder $ \bw_{z} $  and precoder $ \bff_{i} $.  For the sake of simplicity and exposition of fundamental properties, we  select maximum ratio combining (MRC) decoding and maximum ratio transmission (MRT) precoding, i.e.,  $ \bw_{z}=\al_{z}\ba(\bar{\theta}_{z})\ba^{\T}(\bar{\theta}_{z} $) and $ \bff_{i}=\bh_{Bk}+\bH_{BS}\bPhi_{w_{k}}	\bh_{Sk} $, \textcolor{black}{ while the theoretical impact of more sophisticated linear transmission techniques such as regularized zero-forcing (RZF) will be investigated in future work. However, in Sec. V, we have shown the impact of RZF and minimum mean-squared error (MMSE) decoding in terms of simulations.}

	\begin{Theorem}\label{Theorem1}
		Given the PB $\bPhi_{w_{k}}$, the achievable SINR of the received signal at the radar  in a RIS-assisted communication radar coexistence, accounting for  correlated fading, is given by \eqref{SINR10},
		\begin{figure*}
		\begin{align}
			\gamma_{z}=\frac{|\bw_{z}^{\H} \al_{z}\ba(\bar{\theta}_{z})\ba^{\T}(\bar{\theta}_{z} ) \bar{\bu}_{z}|^{2}}{\frac{\bar{\lambda} \rho}{K}(\bw_{z}^{\H}\bR_{\mathrm{R}}\bw_{z}\sum_{i=1}^{K}\tr(\bR_{BSi}\bR_{BS})(\tilde{ \beta}_{BR}+\tilde{ \beta}_{BSR}\tr(\bPhi_{w_{k}}^{\H}\bR_{\mathrm{RIS}}\bPhi_{w_{k}}\bR_{\mathrm{RIS}})))+\sigma^{2}_{r}\|\bw_{z}\|^{2}}\label{SINR10}.
		\end{align}
		\hrulefill
\end{figure*}
		where 
		\begin{align}
			\!\!\!	\bar{\lambda}=\frac{1}{\sum_{i=1}^{K}\!\tr(\bR_{BSi})}.\label{normalization1}
		\end{align}
	\end{Theorem} 
	\proof The proof  is provided in Appendix~\ref{Theor1}.\endproof
	
	\begin{Theorem}\label{Theorem2}
		Given the PB $\bPhi_{w_{k}}$, the achievable SINR of the received signal at UE $k$  in a RIS-assisted communication radar coexistence, accounting for  correlated fading, is given by 
		\begin{align}
			\mathbf{\gamma}_{k}=\frac{S_{l}}{I_{k}}\label{SINR4},
		\end{align}
		where
		\begin{align}
			S_{k}&=\tr^{2}\left(\bR_{BSk}\right)\!,\label{Sk}\\
			I_{k}&=\tr(\bR_{BSk}^{2})-\tr^{2}(\bR_{BSk})+		\sum_{i\ne k}\tr\!\left(\bR_{BSk}\bR_{BSi} \right)\nn\\
			&+\frac{K}{\bar{\lambda} \rho ZL}\sum_{z=1}^{Z}\tr(\bar{\bu}_{z}\bar{\bu}_{z}^{\H}\bR_{RSk})+\frac{\sigma_{c}^{2}K}{\bar{\lambda} \rho}. \label{Ik}
		\end{align}
	\end{Theorem} 
	\proof The proof  is provided in Appendix~\ref{Theor2}.\endproof

	The downlink achievable  sum SE is obtained as
	\begin{align}
		\mathrm{SE}	=\frac{1}{\tau_{\mathrm{c}}}\sum_{k=1}^{K}\log_{2}\left ( 1+\mathbf{\gamma}_{k}\right)\!.\label{LowerBound}
	\end{align}

	\textcolor{black}{Both Theorems 1 and 2 depend on statistical CSI in terms of path losses and correlation matrices. Hence, they are quite advantageous compared to SINRs based on I-CSI, which  would require frequent optimization taking place at each coherence interval.   Both theorems   unveil the impact of spatial correlation on the achievable SINRs, which is discussed in the following remark.}
	\begin{remark}
		\textcolor{black}{The presence of the spatial correlation can enhance the capability of the STAR-RIS to tailor the wireless channel. This is obvious from the direct inspection of the SINRS presented in Theorems 1 and 2. To be more specific, consider the term $ \bPhi_{w_{k}}^{\H}\bR_{\mathrm{RIS}}\bPhi_{w_{k}}\bR_{\mathrm{RIS}} $ as an example. If the spatial correlation is negligible, i.e., $ \bR_{\mathrm{RIS}}=\Id_{N} $, this term becomes fixed and equal to $ \sum_{i=1}^{N}(\beta_{i}^{w_{k}})^{2} $ without the option to be adjusted by the surface controller since the matrix $ \bPhi_{w_{k}} $ is a unitary matrix. Fortunately, correlation appears in practice, which means that that the corresponding terms can be shaped by the surface as will be shown below in Section IV.}''	\end{remark}
	
	\section{Sum SE Maximization Design}\label{SumSEMaximizationDesign}
	In this section, we elaborate on the radar transmit and receive beamforming as well as the PB of the STAR-RIS that optimize the sum SE	of a STAR-RIS-assisted communication radar coexistence system. Based on infinite-resolution phase shifters, the formulation of the  optimization problem providing the maximization of the sum SE is written as
	\begin{subequations}
		\begin{align}
			(\mathcal{P}1)~~&\max_{\bar{\bu}_{z},\bw_{z},
				\thetv,\betv} 	\;		\mathrm{{SE}}
			\label{Maximization1} \\
			&~~~	\mathrm{s.t}~~~~\;\!	\gamma_{z}\ge \gamma^{r}	\label{Maximization2} \\
			&\;\quad\;\;\;\;\;\!\!~~~~\!\textcolor{black}{\sum_{z=1}^{Z}\|\bar{\bu}_{z}\|^{2}}\le P_{\mathrm{max}},\label{Maximization21}\\
			&~~~~~~~~\quad (\beta_{n}^{t})^{2}+(\beta_{n}^{r})^{2}=1,  \forall n \in \mathcal{N}\label{Maximization22}\\
			&~~~~~~~~\quad\beta_{n}^{t}\ge 0, \beta_{n}^{r}\ge 0,~\forall n \in \mathcal{N}\label{Maximization23}\\
			&~~~~~~~~\quad|\theta_{n}^{t}|=|\theta_{n}^{r}|=1, ~\forall n \in \mathcal{N}
			\label{Maximization30} 
		\end{align}
	\end{subequations}
	where we denote $\thetv=[(\thetv^{t})^{\T}, (\thetv^{r})^{\T}]^{\T} $ and $\betv=[(\betv^{t})^{\T}, (\betv^{r})^{\T}]^{\T} $  to achieve a compact description. Also, $ \gamma^{r} $ expresses the SINR threshold. The constraint \eqref{Maximization2} guarantees the radar detection performance for each direction. Constraint \eqref{Maximization21} expresses the radar transmit power limitation with $ P_{\mathrm{max}} $ being the budget of the total radar power consumption. Constraints \eqref{Maximization22} and \eqref{Maximization23} concern the amplitude of the STAR-RIS, while \eqref{Maximization30} denotes the   uni-modulus constraint on the phase shifts.
	
	The  problem $ (\mathcal{P}1) $ is difficult to solve because the objective function is  non-convex, and a coupling among the optimization variables of the STAR-RIS appears, which are the amplitudes and the phase shifts for transmission and reflection. For the solution, we perform alternating optimization by keeping first the radar beamforming vectors fixed while optimizing the STAR-RIS  in an iterative manner until reaching  convergence to a stationary point. Next, the STAR-RIS PBs are considered fixed, and we optimize the radar beamforming vectors. \textcolor{black}{Since the gradient ascent results in local optimum, the overall algorithm, consisted of the two subproblems, results in local optimal.}
	\subsection{Problem Formulation for STAR-RIS}\label{ProblemFormulation2}
	Given the transmit and receive radar beamforming, the formulation problem for  the STAR-RIS is written as
	\begin{subequations}
		\begin{align}
			(\mathcal{P}2)~~&\max_{	\thetv,\betv} 	\;		\mathrm{{SE}}
			\label{Maximization11} \\
			&~~~	\mathrm{s.t}~ (\beta_{n}^{t})^{2}+(\beta_{n}^{r})^{2}=1,  \forall n \in \mathcal{N}\label{Maximization221}\\
			&~~~~~\quad\beta_{n}^{t}\ge 0, \beta_{n}^{r}\ge 0,~\forall n \in \mathcal{N}\label{Maximization2312}\\
			&~~~~~\quad|\theta_{n}^{t}|=|\theta_{n}^{r}|=1, ~\forall n \in \mathcal{N}
			\label{Maximization31} 
		\end{align}
	\end{subequations}
	
	For the sake of exposition, the feasible set of $ 	(\mathcal{P}2) $  can be described by the following  two sets: $ \Theta=\{\thetv\ |\ |\theta_{i}^{t}|=|\theta_{i}^{r}|=1,i=1,2,\ldots N\} $, and $ \mathcal{B}=\{\betv\ |\ (\beta_{i}^{t})^{2}+(\beta_{i}^{r})^{2}=1,\beta_{i}^{t}\geq0,\beta_{i}^{r}\geq0,i=1,2,\ldots N\} $. 
	
	The introduction of STAR-RIS induces further constraints to a conventional RIS, which has only a  unit-modulus constraint. Specifically, the first constraint, based on  the energy conservation law, includes the two types of passive beamforming, namely transmission and reflection beamforming, to be optimized, which are coupled with each other. Also, 
	we observe the non-convexity of problem  $ (\mathcal{P}2) $. However, the projection operators  of the sets $\Theta$ and $ \mathcal{B}$  can be done in closed form, which leads us to apply the  projected  gradient ascent method (PGAM) \cite[Ch. 2]{Bertsekas1999} for the  optimization of $\thetv$ and $\betv$

	We propose a PGAM consisting of the following iterations
	\begin{subequations}\label{mainiteration}\begin{align}
			\thetv^{n+1}&=P_{\Theta}(\thetv^{n}+\mu_{n}\nabla_{\thetv}f(\thetv^{n},\betv^{n})),\label{step1} \\ \betv^{n+1}&=P_{\mathcal{B}}(\betv^{n}+{\mu}_{n}\nabla_{\betv}f(\thetv^{n},\betv^{n})),\label{step2} \end{align}
	\end{subequations}
	where the superscript expresses the iteration count,  $\mu_n$ is the step size for both $\thetv$ and $\betv$, and  $P_{\Theta}(\cdot) $ and $ P_{\mathcal{B}}(\cdot) $ are the projections onto $ \Theta $ and $ \mathcal{B} $, respectively. According to this method, we move from the current iterate $(\thetv^{n},\betv^{n})$  along the gradient direction to increase the objective.

	Concerning the projection onto the set  $ \Theta  $, it is written  for a given $\thetv\in \mathbb{C}^{2N\times 1}$  in terms of the entrywise operation
	\begin{equation}
		P_{\Theta}(\thetv)=\thetv/|\thetv|=e^{j\angle\thetv}.
	\end{equation}
	
	In the case of  $P_{ \mathcal{B} }(\betv)$, we notice that the projection is complicated because of the constraint  $(\beta_{i}^{t})^{2}+(\beta_{i}^{r})^{2}=1,\beta_{i}^{t}\geq0,\beta_{i}^{r}\geq0$. To make the projection more efficient, we assume that $\beta_{i}^{t}$ and $\beta_{i}^{r}$ can take negative values without affecting  the optimality of the proposed solution. The reason is that the objective remains the same despite changing the sign of both $\beta_i^{u}$ and $\theta_{i}^{u}$, $u\in{t,r}$. By projecting $\beta_{i}^{t}$ and $\beta_{i}^{r}$ onto the entire unit circle,   we can obtain $P_{ \mathcal{B} }(\betv)$  as
	\begin{subequations}
		\begin{align}
			\left[\ensuremath{P_{\mathcal{B}}(}\boldsymbol{\beta})\right]_{i} & =\frac{{\beta}_{i}}{\sqrt{{\beta}_{i}^{2}+{\beta}_{i+N}^{2}}},i=1,2,\ldots,N\\
			\left[\ensuremath{P_{\mathcal{B}}(}\boldsymbol{\beta})\right]_{i+N} & =\frac{{\beta}_{i+N}}{\sqrt{{\beta}_{i}^{2}+{\beta}_{i+N}^{2}}}, i=1,2,\ldots,N.
		\end{align}
	\end{subequations}

	The convergence of the proposed PGAM depends on the selection of the step size in \eqref{step1} and \eqref{step2}. Given that it is difficult to find the Lipschitz constant of the corresponding gradient, which would provide the ideal step size, we resort to the application of the Armijo-Goldstein backtracking line search to obtain the step size at each iteration.

	The step size $  \mu_{n} $ in \eqref{mainiteration} can be obtained as $ \mu_{n} = L_{n}\kappa^{m_{n}} $, where $ m_{n} $ is the
	smallest nonnegative integer satisfying
	\begin{align}
		f(\thetv^{n+1},\betv^{n+1})\geq	Q_{L_{n}\kappa^{m_{n}}}(\thetv^{n}, \betv^{n};\thetv^{n+1},\betv^{n+1}),\label{qfunc}
	\end{align}
	where $ L_n>0 $, $ \kappa \in (0,1) $, and $ Q_{\mu}(\thetv, \betv;\bx,\by) $ is a quadratic approximation of $\mathrm{{SE}}$ as
	\begin{align}
		&	Q_{\mu}(\thetv, \betv;\bx,\by)=\mathrm{{SE}}+\langle	\nabla_{\thetv}\mathrm{{SE}},\bx-\thetv\rangle-\frac{1}{\mu}\|\bx-\thetv\|^{2}_{2}\nn\\
		&+\langle\nabla_{\betv}\mathrm{{SE}},\by-\betv\rangle-\frac{1}{\mu}\|\by-\betv\|^{2}_{2}.
	\end{align}   
	Note that \eqref{qfunc} can be performed by an iterative procedure. Also, in the proposed PGAM, we use the step size at iteration $n$ as the initial step size at iteration $n+1$. The method is summarized in Algorithm \ref{Algoa1}. 
	
	\begin{algorithm}[th]
		\caption{Projected Gradient Ascent Algorithm for the STAR-RIS Design\label{Algoa1}}
		\begin{algorithmic}[1]
			\STATE Input: $\thetv^{0},\betv^{0},\mu_{1}>0$, $\kappa\in(0,1)$
			\STATE $n\gets1$
			\REPEAT
			\REPEAT \label{ls:start}
			\STATE $\thetv^{n+1}=P_{\Theta}(\thetv^{n}+\mu_{n}\nabla_{\thetv}f(\thetv^{n},\betv^{n}))$
			\STATE $\betv^{n+1}=P_{B}(\betv^{n}+\mu_{n}\nabla_{\betv}f(\thetv^{n},\betv^{n}))$
			\IF{ $f(\thetv^{n+1},\betv^{n+1})\leq Q_{\mu_{n}}(\thetv^{n},\betv^{n};\thetv^{n+1},\thetv^{n+1})$}
			\STATE $\mu_{n}=\mu_{n}\kappa$
			\ENDIF
			\UNTIL{ $f(\thetv^{n+1},\betv^{n+1})>Q_{\mu_{n}}(\thetv^{n},\betv^{n};\thetv^{n+1},\thetv^{n+1})$}\label{ls:end}
			\STATE $\mu_{n+1}\leftarrow\mu_{n}$
			\STATE $n\leftarrow n+1$
			\UNTIL{ convergence}
			\STATE Output: $\thetv^{n+1},\betv^{n+1}$
		\end{algorithmic}
	\end{algorithm}

	\begin{proposition}\label{LemmaGradients}
		The complex gradients $ \nabla_{\thetv}\mathrm{{SE}} $ and  $\nabla_{\betv}\mathrm{{SE}} $ are obtained in closed forms as
		\begin{subequations}
			\begin{align}
				\nabla_{\thetv}\mathrm{{SE}} &=[\nabla_{\thetv^{t}}\mathrm{{SE}}^{\T}, \nabla_{\thetv^{r}}\mathrm{{SE}}^{\T}]^{\T},\\
				\nabla_{\thetv^{t}}\mathrm{{SE}}&=\frac{\tau_{\mathrm{c}}-\tau}{\tau_{\mathrm{c}}\log2}\sum_{k=1}^{K}\frac{	I_{k}\nabla_{\thetv^{t}}{S_{k}}-S_{k}	\nabla_{\thetv^{t}}{I_{k}}}{(1+\gamma_{k})I_{k}^{2}} ,\\
				\nabla_{\thetv^{r}}\mathrm{{SE}}&=\frac{\tau_{\mathrm{c}}-\tau}{\tau_{\mathrm{c}}\log2}\sum_{k=1}^{K}\frac{	I_{k}\nabla_{\thetv^{r}}{S_{k}}-S_{k}	\nabla_{\thetv^{r}}{I_{k}}}{(1+\gamma_{k})I_{k}^{2}}, 
			\end{align}
		\end{subequations}
		\begin{subequations}
			\begin{align}
				\nabla_{\thetv^{t}}S_{k}&=\begin{cases}
					\nu_{k}\diag\bigl(\mathbf{A}_{t}\diag(\boldsymbol{{\beta}}^{t})\bigr) & w_{k}=t\\
					0 & w_{k}=r
				\end{cases}\label{derivtheta_t}\\
				\nabla_{\thetv^{r}}S_{k}&=\begin{cases}
					\nu_{k}\diag\bigl(\mathbf{A}_{r}\diag(\boldsymbol{{\beta}}^{r})\bigr) & w_{k}=r\\
					0 & w_{k}=t
				\end{cases}\label{derivtheta_r}\\
				\nabla_{\thetv^{t}}I_k &=\diag\bigl(\tilde{\mathbf{A}}_{kt}\diag(\boldsymbol{{\beta}}^{t})\bigr)\label{derivtheta_t_Ik}\\
				\nabla_{\thetv^{r}}I_k &=\diag\bigl(\tilde{\mathbf{A}}_{kr}\diag(\boldsymbol{\beta}^{r})\bigr)\label{derivtheta_r_Ik}
			\end{align}
		\end{subequations}
		with $\bA_{w_k}= \bR_{\mathrm{RIS}} \bPhi_{w_k} \bR_{\mathrm{RIS}} $ for $w_k\in\{t,r\}$, $v_{k}=2\tilde{ \beta}_{BSk}\tr(\bR_{BSk})\tr(\mathbf{R}_{\mathrm{BS}}) $, 
		\begin{equation}
			\tilde{\mathbf{A}}_{ku}=\begin{cases}
				\bigl(\bar{\nu}_{k}+\sum\nolimits _{i\in\mathcal{K}_{u}}^{K}\tilde{\nu}_{ki}\bigr)\mathbf{A}_{u} & w_{k}=u\\
				\sum\nolimits _{i\in\mathcal{K}_{u}}^{K}\tilde{\nu}_{ki}\mathbf{A}_{u} & w_{k}\neq u,
			\end{cases}\label{A_tilde_general}
		\end{equation}
		$u\in\{t,r\}$, $\bar{\nu}_{k}=2\tilde{ \beta}_{BSk}\tr(\bR_{BSk}\mathbf{R}_{\mathrm{BS}})-2\tilde{ \beta}_{BSk}\tr(\bR_{BSk})\tr(\mathbf{R}_{\mathrm{BS}})+\frac{K}{\bar{\lambda} \rho ZL}\tilde{ \beta}_{RSk}\sum_{z=1}^{Z}\tr(\bar{\bu}_{z}\bar{\bu}_{z}^{\H}\bR_{\mathrm{R}})$, 
		$\tilde{\nu}_{ki}=\tilde{ \beta}_{BSk}\tr\bigl(\mathbf{R}_{\mathrm{BS}}(\bR_{BSi}+\bR_{BSk})\bigr)-\frac{K}{\bar{\lambda}^{2} \rho}\tilde{ \beta}_{BSk}\tr(\mathbf{R}_{\mathrm{BS}})(\frac{1}{ZL}\sum_{z=1}^{Z}\tr(\bar{\bu}_{z}\bar{\bu}_{z}^{\H}\bR_{RSk})+\sigma_{c}^{2})$.
		
		In a similar way, we obtain the  gradient $\nabla_{\betv}\mathrm{{SE}} $ as
		\begin{subequations}\label{eq:deriv:wholebeta}
			\begin{align}
				\nabla_{\betv}\mathrm{{SE}} &=[\nabla_{\betv^{t}}\mathrm{{SE}}^{\T}, \nabla_{\betv^{r}}\mathrm{{SE}}^{\T}]^{\T},\\
				\nabla_{\betv^{t}}\mathrm{{SE}}&=\frac{\tau_{\mathrm{c}}-\tau}{\tau_{\mathrm{c}}\log2}\sum_{k=1}^{K}\frac{	I_{k}\nabla_{\betv^{t}}{S_{k}}-S_{k}	\nabla_{\betv^{t}}{I_{k}}}{(1+\gamma_{k})I_{k}^{2}}, \label{gradbetat:final}\\
				\nabla_{\betv^{r}}\mathrm{{SE}}&=\frac{\tau_{\mathrm{c}}-\tau}{\tau_{\mathrm{c}}\log2}\sum_{k=1}^{K}\frac{	I_{k}\nabla_{\betv^{r}}{S_{k}}-S_{k}	\nabla_{\betv^{r}}{I_{k}}}{(1+\gamma_{k})I_{k}^{2}} ,
			\end{align}
		\end{subequations} 
		where
		\begin{subequations}
			\begin{align}
				\nabla_{\betv^{t}}S_{k}&=\begin{cases}
					2\nu_k\Re\bigl\{\diag\bigl(\mathbf{A}_{k}\herm\diag(\btheta^{t})\bigr)\bigr\} & w_{k}=t\\
					0 & w_{k}=r
				\end{cases}\label{derivbeta_t}\\
				\nabla_{\betv^{r}}S_{k}&=\begin{cases}
					2\nu_k\Re\bigl\{\diag\bigl(\mathbf{A}_{k}\herm\diag(\btheta^{r})\bigr)\bigr\} & w_{k}=r\\
					0 & w_{k}=t
				\end{cases}\label{derivbeta_r}\\
				\nabla_{\betv^{t}}I_k &=2\Re\bigl\{\diag\bigl(\tilde{\mathbf{A}}_{kt}\herm\diag(\boldsymbol{\btheta}^{t})\bigr)\bigr\}\\
				\nabla_{\betv^{r}}I_k &=2\Re\bigl\{\diag\bigl(\tilde{\mathbf{A}}_{kr}\herm\diag(\boldsymbol{\btheta}^{r})\bigr)\bigr\}.
			\end{align}
		\end{subequations}
		
	\end{proposition}

	\begin{proof}
		Please see Appendix~\ref{lem2}.	
	\end{proof}
	
	\subsubsection*{Complexity Analysis of Algorithm \ref{Algoa1}} Herein, we focus on the complexity for each iteration of Algorithm \ref{Algoa1}, which  includes the objective and its gradient value. In the case of  the objective value, we observe that $ \bR_{BSk} $ and $ \bR_{RSk} $ require   $O(M^{2}+N^2)$ complex multiplications since $O(M^{2}+N^2)$ multiplications are required for obtaining  $\diag(\bA_{w_k})=\diag(\bR_{\mathrm{RIS}} \bPhi_{w_{k}} \bR_{\mathrm{RIS}})$  given that $  \bPhi_{w_{k}} $ is diagonal, and $O(N^2)$ additional  multiplications are required to evaluate $\tr(\bA_{w_k} \bPhi_{w_{k}}^{\H})\bR_{\mathrm{BS}}$.   The complexity of $ \tr(\bar{\bu}_{z}\bar{\bu}_{z}^{\H}\bR_{RSk}) $ is $O(Q^2)$. Hence, the  overall complexity of the objective is $O(K(N^2+M^2+Q^2))$. Regarding the computation of the gradients $\nabla_{\thetv}f(\thetv,\betv)$ and $\nabla_{\betv}f(\thetv,\betv)$, the traces require $O(M^{2}+N^{2})$ complex multiplications, while $\mathbf{A}_{t}\diag(\boldsymbol{{\beta}}^{t})$ requires $O(N)$ multiplications since both matrices are diagonal. Thus, the  complexity of the computation of the gradients for each iteration is $O(K(N^2+M^2+Q^2))$, which is also the complexity of the objective. 
	
\textcolor{black}{	Although this work does not focus on the analysis under discrete phase shifts, herein, we only propose a quantization scheme to supplement Algorithm \ref{Algoa1} in the case, where each element of the STAR-RIS can have a finite number of discrete values. Specifically, we denote $ b $  the
	number of bits that correspond to  the phase shift levels $ Q $, where 	$ Q = 2^{b} $. For the sake of simplicity, we assume that the interval $ [0, 2\pi) $ is uniformly quantized as
	\begin{align}
		\mathcal{F}=\{0, \Delta \theta, \ldots, (Q-1)\Delta \theta\},
	\end{align}	
where $ \Delta \theta=2\pi/ Q $. The solution of 	$ (\mathcal{P}2) $  under the constraint $ \theta_{n}^{t}, \theta_{n}^{r} \in \mathcal{F}$, $ \forall n $ results by quantizing the diagonal entries of $ \bTheta $, which is the solution of Algorithm \eqref{Algoa1} to the nearest discrete values in $ \mathcal{F} $. The proposed solution with discrete phase shifts is mentioned as Algorithm $ 1 $-Quantized, and is depicted in numerical results.}
	\subsection{Problem Formulation for Radar}\label{ProblemFormulation3}
	Herein, we focus on the optimal radar beamforming design with given STAR-RIS passive beamforming vectors. In particular, the maximization problem in $ 	(\mathcal{P}1) $ concerns the design of $ \bar{\bu}_{z} $ and $ \bw_{z} $ to avoid interference from the radar to the
	communication receiver. The corresponding optimization problem is formulated as
	\begin{subequations}
		\begin{align}
			(\mathcal{P}3)~~&\max_{\bar{\bu}_{z},\bw_{z}} 	\;		\mathrm{{SE}}
			\label{Maximization3} \\
			&~~	\mathrm{s.t}~~\;\!	\gamma_{z}\ge \gamma^{r}	\label{Maximization2313} \\
			&\;~~~~~~\!\textcolor{black}{\sum_{z=1}^{Z}\|\bar{\bu}_{z}\|^{2}}\le P_{\mathrm{max}},\label{Maximization213}	
		\end{align}
	\end{subequations}
	
	The solution to this problem is obtained  by the following proposition.
	\begin{proposition}\label{Proposition1}
		Under given PB regarding the STAR-RIS, the optimal solution to problem $ 	(\mathcal{P}3) $ is given by
		\begin{align}
			\bw_{z}^{\star}&=(\sigma^{2}_{r}\Id_{M}+\bA)^{-1} \al_{z}\ba(\bar{\theta}_{z})\ba^{\T}(\bar{\theta}_{z} ) \bar{\bu}_{z},\\
			\bar{\bu}_{z}^{\star}&\!=\!\left\{
			\begin{array}{ll}\!\!\!
				\sqrt{\bar{\gamma}^{r}}\ba^{*}(\bar{\theta}_{z})\\	\!\!\!-\frac{\ba^{\T}(\bar{\theta}_{z})\bR_{RSk}\ba^{*}(\bar{\theta}_{z})\bee_{z}}{\bee_{z}^{\H}\bR_{RSk}\bee_{z}+\bar{\lambda}^{\star}}\sqrt{ \bar{\gamma}^{r}},&  \!\!\!\mathrm{if}~\bee_{z}^{\H}\bR_{RSk}\bee_{z}\ne 0,\\
				\sqrt{\bar{\gamma}^{r}}\ba^{*}(\bar{\theta}_{z}), & \mathrm{if}~\bee_{z}^{\H}\bR_{RSk}\bee_{z}=0,
			\end{array} 
			\right. 
		\end{align}
		where $ \bar{\gamma}^{r}=\frac{\gamma^{r}}{|\al_{z}|^{2}}\Big(1-\frac{\ba(\bar{\theta}_{z})^{\H}\bA \ba(\bar{\theta}_{z})}{\sigma^{2}_{r}+\bA}\Big)^{-1} $, $ \bee_{z}=\frac{\bh_{RSk}-(\ba^{\T}(\bar{\theta}_{z})\bh_{RSk})\ba^{*}(\bar{\theta}_{z})}{\|\bh_{RSk}-(\ba^{\T}(\bar{\theta}_{z})\bh_{RSk})\ba^{*}(\bar{\theta}_{z})\|} $, and $ \bar{\lambda}^{\star} $ is the optimal
		Lagrange multiplier for constraint \eqref{Maximization213} that satisfies $ \bar{\lambda}^{\star}\left(\sum_{z=1}^{Z}\|\bar{\bu}_{z}\|^{2}- P_{\mathrm{max}}\right)=0 $.
		
	\end{proposition}
	\proof The proof  is provided in Appendix~\ref{Prop1}.\endproof
	
\begin{remark}		
	\textcolor{black}{Proposition \ref{Proposition1} provides the design of the optimal radar beamforming vector with a given PB. Specifically, the vectors $ \ba^{*}(\bar{\theta}_{z}) $ and $ \bee_{z} $ correspond together to the radar beamforming for direction $ \bar{\theta}_{z} $, where $ \bee_{z} $ is orthogonal to  $ \ba^{*}(\bar{\theta}_{z}) $ and is related to  the interference channel vector $ \bh_{RSk} $. The powers allocated to $ \ba^{*}(\bar{\theta}_{z}) $ and $ \bee_{z} $ are used for controlling the radar SINR and  for reducing the interference to UE $ k $ (communication receiver). Thus, when the radar transmit power is enough large, it can be allocated to   $ \ba^{*}(\bar{\theta}_{z}) $ to satisfy the radar SINR constraint and to   $ \bee_{z} $ to circumvent the interference to the communication receiver being UE $ k $, which means that the communication SINR is not affected by the radar system. }
	\end{remark}

	Hence, the maximum communication SINR  is given by \eqref{maxsinr}, which reveals the relationship between the communication SINR and the radar power budget via the following proposition.
	\begin{figure*}
		\begin{align}
			\gamma_{k}=	\frac{\tr^{2}\left(\bR_{BSk}\right)}{\tr(\bR_{BSk}^{2})-\tr^{2}(\bR_{BSk})+		\sum_{i\ne k}\tr\!\left(\bR_{BSk}\bR_{BSi} \right)+\frac{K}{\lambda \rho ZL}\sum_{z=1}^{Z}\tr(\bar{\bu}_{z}\bar{\bu}_{z}^{\H}\bR_{RSk})+\frac{\sigma_{c}^{2}K}{\bar{\lambda} \rho}}\label{maxsinr}.
		\end{align}
		\hrulefill
	\end{figure*}
	\begin{proposition}\label{Proposition2}
		The communication SINR is increased with the radar power budget since the interference from the radar to UE $ k$ is decreased as the power increases. Specifically, in the case of a sufficiently large power budget, the interference becomes zero, and the corresponding SINR becomes
		\begin{align}
			\gamma_{k}\!=\!	\frac{\tr^{2}\left(\bR_{BSk}\right)}{\tr(\bR_{BSk}^{2})\!-\!\tr^{2}(\bR_{BSk})\!+\!		\sum_{i\ne k}\tr\!\left(\bR_{BSk}\bR_{BSi} \right)\!+\!\frac{\sigma_{c}^{2}K}{\bar{\lambda} \rho}}\label{maxsinr1}.
		\end{align} 
	\end{proposition}
	\proof The proof  is provided in Appendix~\ref{Prop2}.\endproof

		\begin{algorithm}[th]
		\caption{Joint STAR-RIS and Radar Design\label{Algoa3}}
		\begin{algorithmic}[1]
			\STATE Input: $\thetv,\betv$
								\REPEAT \label{ls:start}
			\STATE Update $ \thetv $, $ \betv $ by Proposition 1 in Section IV.A
			\STATE Update  $\bw_{z}$, $ \bar{\bu}_{z}$ by Proposition 2 in  Section IV.B
										\UNTIL{ convergence}
			\STATE Output: $\thetv^{n+1},\betv^{n+1}$, $\bw_{z}^{\star}$, $ \bar{\bu}_{z}^{\star} $.
		\end{algorithmic}
	\end{algorithm}

	\section{Numerical Results}\label{Numerical}
	In this section, we elaborate on the numerical results of the  STAR-RIS-assisted communication radar coexistence system  in terms of analytical results and Monte-Carlo (MC) simulations with $ 10^{3} $ independent channel realizations. 
	\subsection{Simulation Setup}
	The  setup includes a STAR-RIS deployed with a UPA of $ N=64 $ elements that assist the communication between a BS and $ K = 4 $ UEs. The BS is equipped with a uniform linear array (ULA) of $ M =64$ antennas. The  numbers of transmit and receive antennas for the radar are $ Q=12 $. The $xy-$coordinates of the BS and RIS are given as $(x_B,~ y_B) = (0,~0)$ and $(x_S,~ y_S)=(50,~ 10)$, respectively, all in meter units.  The radar is located at  $(x_R,~ y_R)=(30,~ D)$, where $ D=20 $. Also, UEs in $r$ region are located on a straight line between $(x_S-\frac{1}{2}d_0,~y_S-\frac{1}{2}d_0)$ and $(x_S+\frac{1}{2}d_0,~y_S-\frac{1}{2}d_0)$ with equal distances between each two adjacent users, and $d_0 = 20$~m in our simulations. In a similar way, UEs in the $t$ region are located between $(x_S-\frac{1}{2}d_0,~y_S+\frac{1}{2}d_0)$ and $(x_S+\frac{1}{2}d_0,~y_S+\frac{1}{2}d_0)$. \textcolor{black}{Although the simulation setup regarding the  division of UEs into two lines looks simplistic, the line could correspond to a pavement or a road, which is a realistic real-world application.} The dimensions of each RIS element are $ d_{\mathrm{H}}\!=\!d_{\mathrm{V}}\!=\!\lambda/4 $. We consider a system operating
	at a carrier frequency of $ 6~\mathrm{GHz} $.  Moreover, we consider  distance-based path-loss, i.e., the channel gain of a given link $j$ is $ A d_j^{-\alpha_{j}}$, where $A$ is the area of each reflecting element at the RIS, and $\alpha_{j}$ is the path-loss exponent. The same value is assumed for radar-related links, while for  UEs located behind the surface, we assume an additional penetration loss equal to $ 15 $ dB \textcolor{black}{\cite{Kammoun2020,Papazafeiropoulos2021}}. The correlation matrices $ \bR_{\mathrm{BS}}$ and $\bR_{\mathrm{R}} $ are  obtained according to Sec. II.C.
	  Note that the transmit power from the BS is $ \rho=0.1~\mathrm{W} $ and the noise variance is $ \sigma^2=-174+10\log_{10}B_{\mathrm{c}} $, where $B_{\mathrm{c}}=200~\mathrm{kHz}$ is the bandwidth. \textcolor{black}{Similar to \cite{He2022}}, regarding the radar, we assume $ 8  $ detection directions, which are $ -\frac{\pi}{3},-\frac{\pi}{4}, \ldots, \frac{\pi}{4} $, and the PRI is $ L=10 $. In addition, the ratio of $ |\al_{z}|^{2}  $ to noise power is $ \frac{|\al_{z}|^{2}}{\sigma^2}=-12~\mathrm{dB}, \forall z $. We assume that the
	total transmit power for $ 8 $ directions is set as $ 10~\mathrm{W} $ and the SINR requirement is $ 10~\mathrm{dB} $ for each direction.
	
	\subsection{Benchmark Schemes}
	To illustrate the performance comparison,  we present the following benchmark schemes.
	
	\begin{itemize}
		\item We consider the conventional RIS that consists of transmitting-only or reflecting-only elements, each with $ N_{t} $ and $ N_{r} $ elements, such that $ N_{t}+N_{r} =N$. It has to be mentioned that this scheme is equivalent to the MS protocol, where the first $ N_{t} $ elements operate in transmission mode and the remaining $ N_{r} $ elements operate in reflection mode. 
		\item We assume random phase shifts, i.e., the phase shifts of reflecting elements are randomly generated by obeying the Uniform
		distribution within $[0, 2\pi)  $.
		\item We consider a conventional system with no RIS.
	\end{itemize}

	\begin{figure}[!h]
		\begin{center}
			\includegraphics[width=0.8\linewidth]{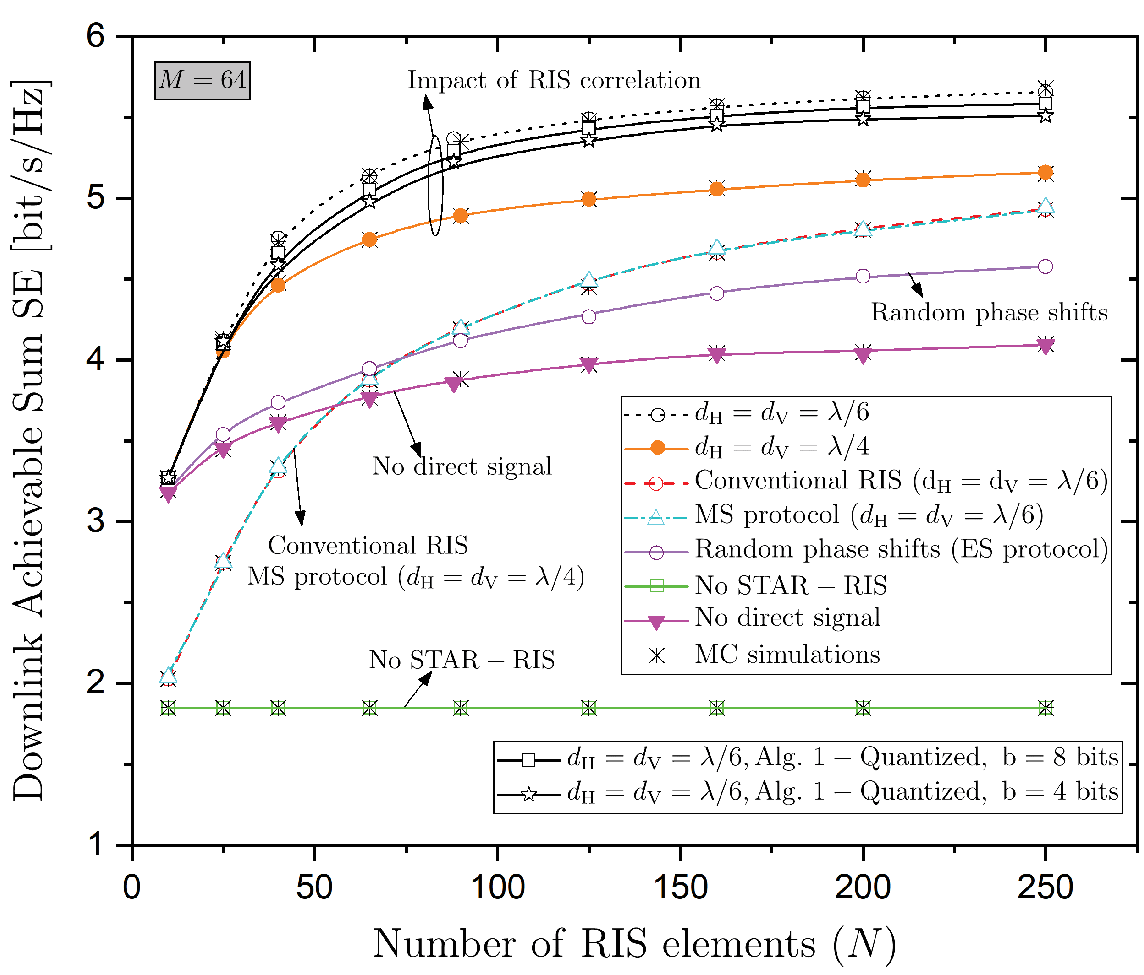}
			\caption{{Downlink achievable communication sum SE versus the number of RIS elements antennas $N$  for varying conditions (Analytical results and MC simulations). }}
			\label{fig2}
		\end{center}
	\end{figure}
	Fig. \ref{fig2} depicts the achievable communication sum SE versus the number of  STAR-RIS elements $ N $. Also, we show the effect of spatial correlation by varying  the size of each RIS element. The first thing that we can easily observe is that  the downlink sum SE increases with $ N $ as expected. In particular, the curves coincide for low $ N $, but the gap increases as  $N$ increases. Regarding the spatial correlation at the STAR-RIS, we observe that  the performance decreases as  the correlation increases. Specifically, as the size of each RIS element decreases, the correlation decreases, and the  sum SE increases. In addition, the MS protocol results in a lower performance since it is a special case of the ES protocol. For the sake of comparison, we have shown the cases of no RIS and blocked direct signal that reveal that the presence of the surface contributes to the system performance since the corresponding lines are lower than the line with a STAR-RIS. Furthermore, we have shown the performance of the conventional RIS, but again the performance is lower because fewer degrees of freedom can be taken into an advantage. \textcolor{black}{Also, the performance of Algorithm 1- Quantized is studied in the case of a STAR-RIS with $4 $-bit and $8 $-bit phase shifters. As can be seen, limiting the surface elements to only discrete phase shifts as opposed to continuous phase shifts decreases the performance reduces the performance. The performance loss is greater for small STAR-RIS resolution, i.e., for $ b=4 $ bits.}

	Fig. \ref{fig3} shows the achievable communication sum SE versus the number of BS antennas $ M $. As can be seen, the sum SE  increases with $ M $. Also, by increasing the size of the RIS elements, we result in an increased surface correlation that leads to worse system performance. If no RIS is considered, no surface optimization can be performed. Similarly, if independent Rayleigh fading is assumed, no surface optimization can be achieved. If we assume random phase shifts, the sum SE is again lower. In addition, the ES protocol results in better performance compared to the MS protocol. However, ES comes with higher complexity. Furthermore,  we have included the baseline scenario with a conventional RIS, where it is shown that the performance is lower compared to the STAR-RIS case.	\textcolor{black}{In addition, for further comparison, we compare the proposed scheme based on statistical CSI with  I-CSI. Despite that I-CSI  results in higher sum-rate, it results in larger overhead. The figure shows the superiority of the I-CSI design. Specifically, we maximize the sum-rate expression using the the algorithm  described in \cite{Peng2021} with the same parameters. 		 We note that that the statistical-CSI based design is preferable because of its lower overhead and because it provides closed-form expressions.}
	
	\begin{figure}[!h]
		\begin{center}
			\includegraphics[width=0.8\linewidth]{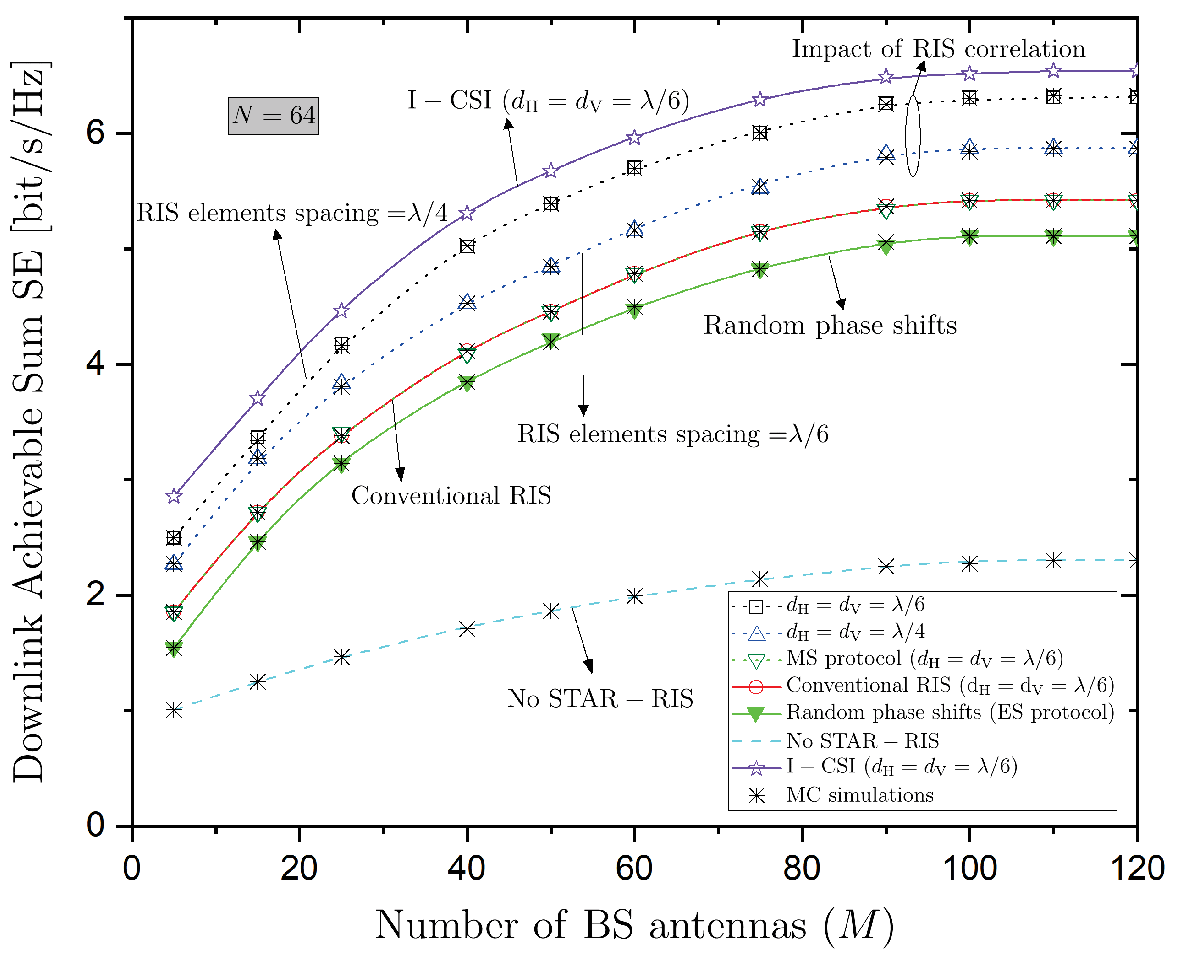}
			\caption{Downlink achievable communication sum SE versus the number of BS antennas $ M $  for varying conditions (Analytical results and MC simulations).}
			\label{fig3}
		\end{center}
	\end{figure}

	\textcolor{black}{In Fig. \ref{fig10}, we depict the achievable communication sum SE versus the number of radar antennas $ Q $ in terms of simulations. We observe that the sum SE increases with $ Q $. Moreover, we observe how more sophisticated precoders and decoders perform. Specifically, the application of RZF, MMSE results in higher performance than MRT, MRC as expected. Also, we have plotted the cases of $ \lambda/4 $ and  $ \lambda/2 $, where the former exhibits better performance.}
	
	\begin{figure}[!h]
		\begin{center}
			\includegraphics[width=0.8\linewidth]{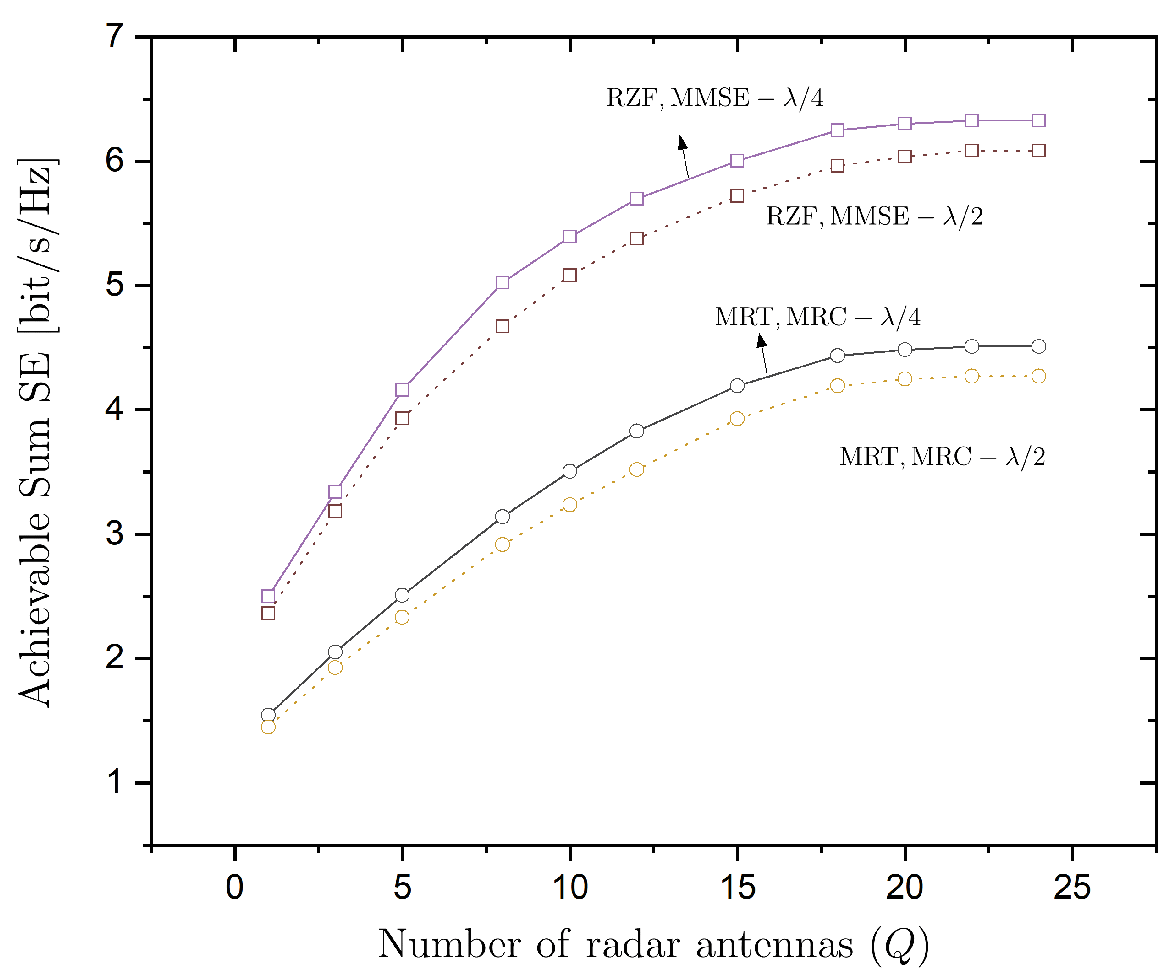}
			\caption{\textcolor{black}{Downlink achievable communication sum SE versus the number of radar antennas $Q$  for different precoders/decoders and element spacing (MC simulations).} }
			\label{fig10}
		\end{center}
	\end{figure}

	Fig. \ref{fig4} illustrates the achievable communication sum SE versus the SNR for $N=200$   (solid lines) and $N=40$ (dotted lines) elements. When $N=40$, the performance is worse. Also, for low SNR, ES and MS protocols present a small gap.  For high SNR, the  gap increases with increasing SNR. In the case of a conventional RIS, the performance is lower than STAR-RIS. In a similar way, the case of  independent Rayleigh fading appears a lower performance. \textcolor{black}{ The reasons for these observations can be justified as follows.  At low SNR, it is more beneficial to focus on UEs in the reflection region as they are closer to the BS. In particular, this is confirmed by the fact that, after running the proposed algorithm, we observe $\beta_{n}^{r}\approx 1,\forall n\in \mathcal{N}$. As a result, the performances of the ES and MS protocols, as well as the conventional RIS, are almost the same. However, as SNR increases, the increase in the sum SE becomes minimal if we continue to focus on UEs in the reflection region. Thus, at high SNR, directing some power to UEs in the transmission region can improve the total SE. This leads to performance differences between the ES and MS protocols, as well as the conventional RIS.}
	
	\begin{figure}[!h]
		\begin{center}
			\includegraphics[width=0.8\linewidth]{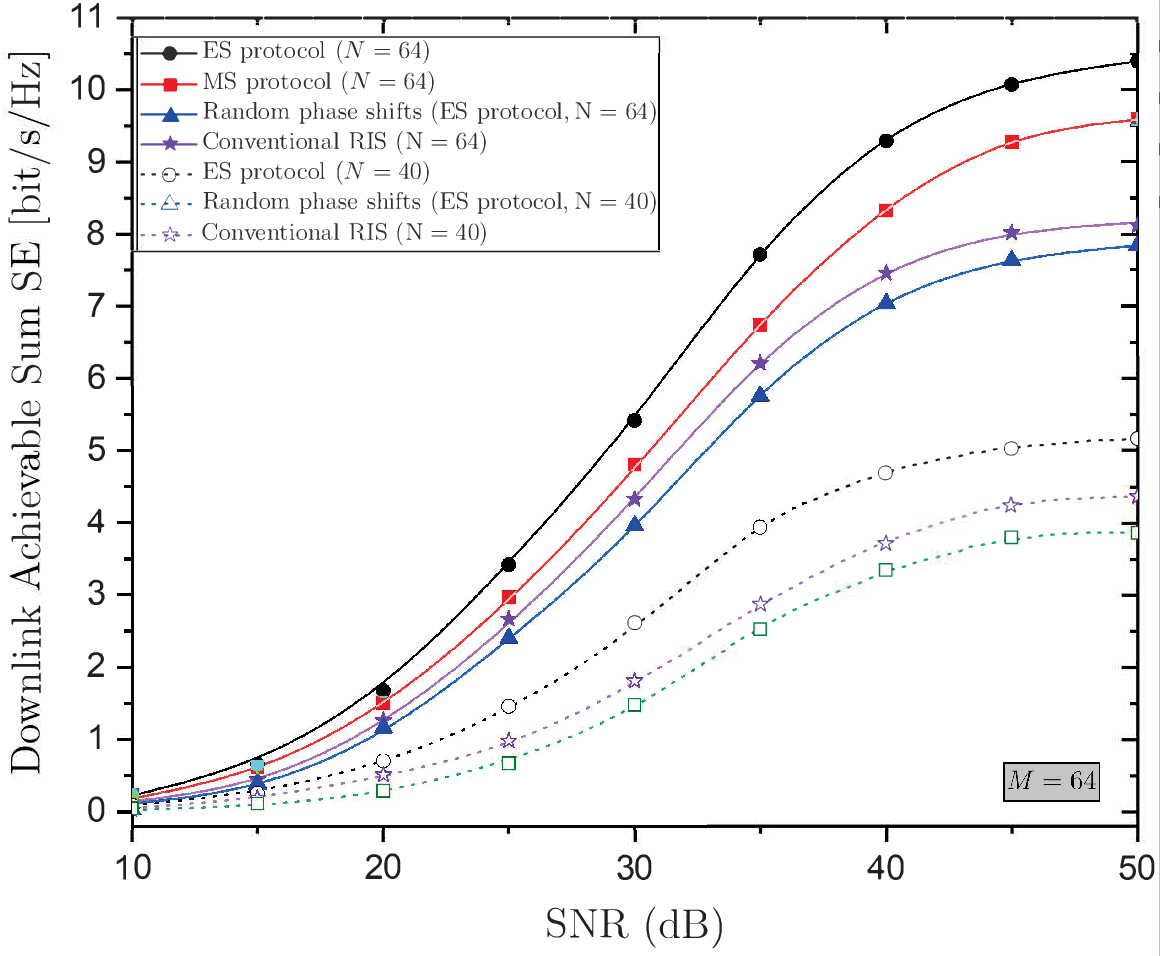}
			\caption{{Downlink achievable communication sum SE versus the SNR  for varying conditions. }}
			\label{fig4}
		\end{center}
	\end{figure}
	
	In Fig. \ref{fig6}, we depict the achievable communication sum SE versus the radar SINR requirement, which is $  \gamma^{r} $. It is shown that the SE is almost constant for low values of $  \gamma^{r} $, while it decreases with $  \gamma^{r} $ when it is high. The reason is that if $  \gamma^{r} $ is low, the radar transmit power is sufficient and the interference can be mitigated. When $ \gamma^{r} $ increases, a higher percentage of the radar transmit power is used to keep the radar SINR at a given level, and less power  is used to mitigate the interference from the radar to UE $ k $. As a result, the communication SE is decreased. For the sake of comparison, in the same figure, we have shown the scenarios of conventional RIS, random phase shifts, and no direct signal.
	
	\begin{figure}[!h]
		\begin{center}
			\includegraphics[width=0.8\linewidth]{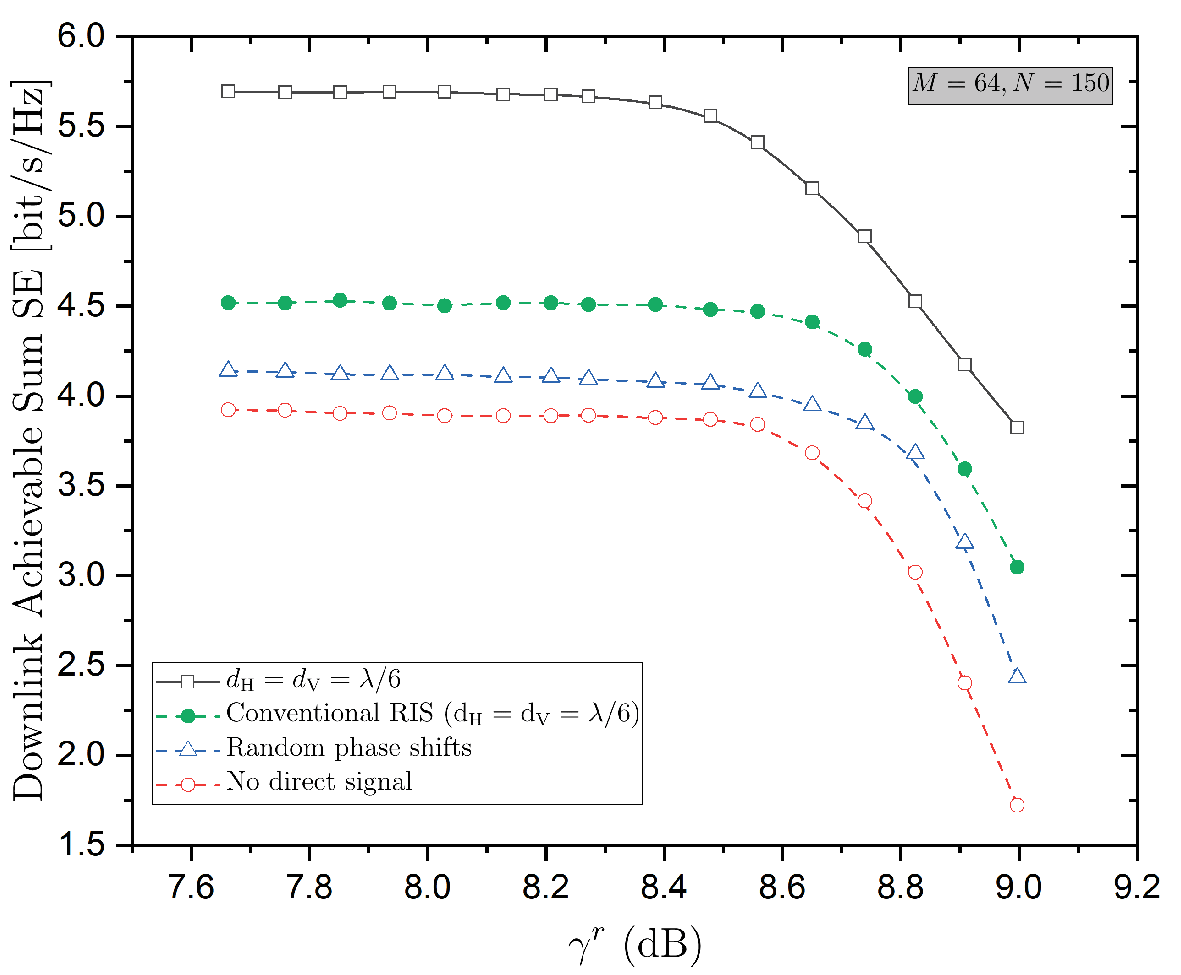}
			\caption{{Downlink achievable communication sum SE versus the radar SINR requirement for varying conditions. }}
			\label{fig6}
		\end{center}
	\end{figure}
	
	In Fig. \ref{fig5}, we show the convergence of the proposed projected gradient algorithm by plotting the achievable communication sum SE versus the iteration count returned by Algorithm \ref{Algoa1} for 5 different randomly generated initial points. Note that the initial values for $\betv_{r}$ and $\betv_{t}$ are $\sqrt{0.5}$, which means equal power splitting between the transmission and reception modes. Also, the initial values for $\thetv_{r}$ and  $\thetv_{r}$ are obtained from the Uniform distribution over $[0,2\pi)$. When the  increase of the objective between the two last iterations is less than $10^{-5}$ or the number of iterations is larger than $200$, we  terminate Algorithm \ref{Algoa1}. In particular, since the proposed  projected gradient algorithm  provides the solution to a nonconvex problem, different initial points converge to different points  as shown in the figure. To reduce this performance sensitivity of Algorithm  \ref{Algoa1} with respect to the initial points, we run it from different initial points and select the best convergent solutions. Simulations have shown that a balanced trade-off between    the  sum SE and complexity is obtained by running Algorithm  \ref{Algoa1} from 5 randomly generated initial points.
	
	\begin{figure}[!h]
		\begin{center}
			\includegraphics[width=0.8\linewidth]{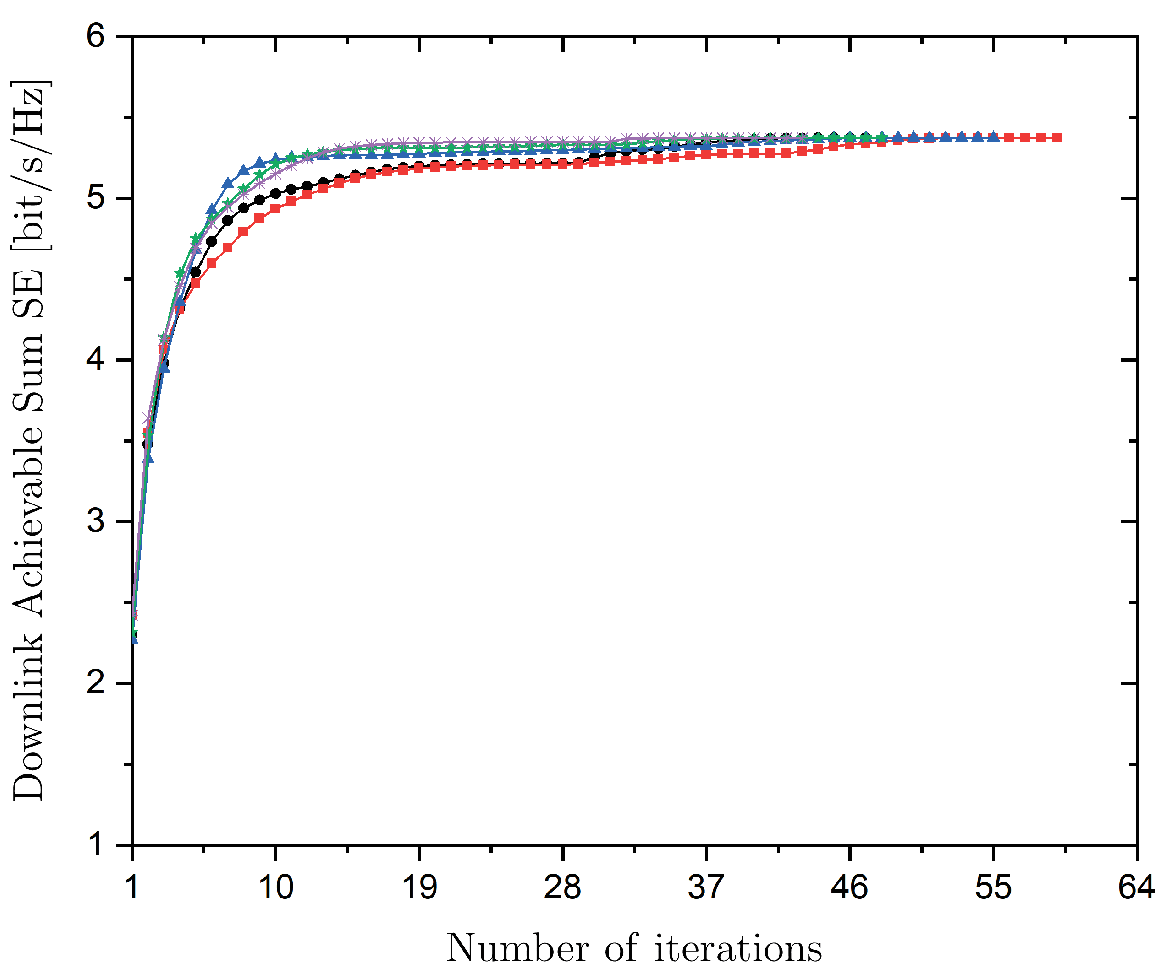}
			\caption{Convergence of Algorithm \ref{Algoa1} for a STAR-RIS-assisted communication radar coexistence system   ($M=64$, $ N=100 $, $Q=12$, $K=4 $). }
			\label{fig5}
		\end{center} 
	\end{figure}

	\section{Conclusion} \label{Conclusion} 
	In this work, we studied a STAR-RIS-assisted radar coexistence system to suppress the mutual interference while boosting the communication performance and providing full space coverage. Note that a STAR-RIS is more robust than a conventional reflecting-only RIS in terms of both performance and coverage. Specifically, we formulated a communication SINR maximization problem subject to radar SINR  and  radar transmit power constraints. Based on statistical CSI and by developing an AO optimization algorithm, we optimized both the STAR-RIS and the radar beamforming. In particular, the amplitudes and phase shifts of the surface were optimized simultaneously, which reduces the overhead. This management is crucial in STAR-RIS-assisted systems, which require the double number of variables compared to conventional RIS. Numerical results investigated the impact of the proposed architecture by various parameters, and its impact on benchmark schemes.  \textcolor{black}{As a benchmark, we devised a design based on  I-CSI, and we showed that the proposed statistical CSI based scheme is inferior in terms of the sum rate but superior in terms of required overhead.}

	\begin{appendices}
		\section{Proof of Theorem~\ref{Theorem1}}\label{Theor1}
		The proof requires the derivation of the first in the denominator of \eqref{SINR1}. Specifically, we have
		\begin{align}
&			\EE \{|(\bH_{BR}^{\H}+\bH_{SR}^{\H}\bPhi_{w_{k}}^{\H}\bH_{BS}^{\H}) \sum_{i=1}^{K}\bff_{i}|^{2}\}=\EE \{|\bH_{BR}^{\H} \sum_{i=1}^{K}\bff_{i}|^{2}\}\nn\\
&+\EE \{|\bH_{SR}^{\H}\bPhi_{w_{k}}^{\H}\bH_{BS}^{\H} \sum_{i=1}^{K}\bff_{i}|^{2}\},	\label{Theor11}		
		\end{align}
		where we have applied the property $\EE\left\{ |X+Y|^{2}\right\} =\EE\left\{ |X|^{2}\right\} +\EE\left\{ |Y^{2}|\right\}$, which is valid for any two uncorrelated random variables while one of them has zero mean value.
		The first term of \eqref{Theor11} is obtained as
		\begin{align}
&			\EE \{|\bH_{BR}^{\H} \sum_{i=1}^{K}\bff_{i}|^{2}\}
			\nn\\
			&= \sum_{i=1}^{K}\tilde{ \beta}_{BR}\	\EE \{\bR_{\mathrm{R}}^{1/2}\bZ_{BR}^{\H}\bR_{\mathrm{BS}}^{1/2}\bR_{BSi}\bR_{\mathrm{BS}}^{1/2}\bZ_{BR}\bR_{\mathrm{R}}^{1/2}\}\\
			&= \tilde{ \beta}_{BR}\sum_{i=1}^{K}\tr(\bR_{BSi}\bR_{BS})\bR_{\mathrm{R}},\label{term11}
		\end{align}
		where we first have obtained the conditional expectation with respect to $ \bH_{BSi} $. Next, we used the property $ \EE\{\bV \bU\bV^{\H}\} =\tr (\bU) \Id_{M}$ with $\bU  $ being a deterministic square matrix, and $ \bV $ being any matrix with independent and identically distributed (i.i.d.) entries of zero mean and unit variance.
		Similarly, the second term of \eqref{Theor11} can be written as
		\begin{align}
		&	\EE \{|\bH_{SR}^{\H}\bPhi_{w_{k}}^{\H}\bH_{BS}^{\H} \sum_{i=1}^{K}\bff_{i}|^{2}\}\nn\\
			&=\sum_{i=1}^{K}\EE \{|\bH_{SR}^{\H}\bPhi_{w_{k}}^{\H}\bH_{BS}^{\H}\bR_{BSi}\bH_{BS}\bPhi_{w_{k}}\bH_{SR}\},\\
			&=\tilde{ \beta}_{BS}\sum_{i=1}^{K}\EE \{|\tr(\bR_{BSi}\bR_{\mathrm{BS}}) \bH_{SR}^{\H}\bPhi_{w_{k}}^{\H}\bR_{\mathrm{RIS}}\bPhi_{w_{k}}\bH_{SR}\},\\
			&=\tilde{ \beta}_{BSR}\sum_{i=1}^{K}\tr(\bR_{BSi}\bR_{\mathrm{BS}})\tr(\bPhi_{w_{k}}^{\H}\bR_{\mathrm{RIS}}\bPhi_{w_{k}}\bR_{\mathrm{RIS}})\bR_{\mathrm{R}},\label{term21}
		\end{align}
		where $\tilde{ \beta}_{SR}= \tilde{ \beta}_{BSR}\tilde{ \beta}_{BS} $.
		The normalization parameter is obtained as
		\begin{align}
			\!\!\!	\bar{\lambda}&=\frac{1}{\sum_{i=1}^{K}\!\EE\{\mathbf{f}_{i}^{\H}\mathbf{f}_{i}\}}=\frac{1}{\sum_{i=1}^{K}\!\mathbb{E}\{\mathbf{h}_{BSi}^{\H}\mathbf{h}_{BSi}\}}\nn\\
			&=\frac{1}{\sum_{i=1}^{K}\!\tr(\bR_{BSi})}.\label{normalization}
		\end{align}
		By substituting \eqref{term11}, \eqref{term21}, and \eqref{normalization} in \eqref{SINR1}, we conclude the proof.
		
		\section{Proof of Theorem~\ref{Theorem2}}\label{Theor2}
		The derivation of 	\eqref{SINR2} starts by focusing on its numerator, which can be written as
		\begin{align}
			|\EE\{\bh_{BSk}^{\H}\bff_{k}\}|^{2}&=|\tr\big( \EE\{\bh_{BSk} \bh_{BSk}^{\H} \} \!\big)|^{2} \label{term02}\\
			&=|\tr\left(\bR_{BSk}\right)\!|^{2}\label{term2},
		\end{align}
		where, in\eqref{term02}, we have applied  the property  $\bx^{\H}\by = \tr(\by \bx^{\H})$, while the last  equation is derived after computing the expectation.
		
		Next, the first term in the denominator of \eqref{SINR2} is obtained as
		\begin{align}
		&	\EE\big\{ \big|\bh_{BSk}^{\H}\bff_{k}-		\EE\{\bh_{BSk}^{\H}\bff_{k}\}\big|^{2}\big\}\!
			\nn\\
			&=\!
			\EE\big\{ \big| \bh_{BSk}^{\H}\bh_{BSk}\big|^{2}\big\}\!-\!\big|\EE\big\{
			\bh_{BSk}^{\H}\bh_{BSk}\big\}\big|^{2}\label{est2}\\
			&=\tr(\bR_{BSk}^{2})-\tr^{2}(\bR_{BSk}) \label{est3},
		\end{align}
		where in~\eqref{est2}, we have applied the property $\EE\left\{ |X+Y|^{2}\right\} =\EE\left\{ |X|^{2}\right\} +\EE\left\{ |Y^{2}|\right\}$. 
		
		The second term of the denominator in \eqref{SINR2} is written as
		\begin{align}
			&\EE\big\{ \big| \bh_{BSk}^{\H}\bff_{i}\big|^{2}\big\}=\tr\!\left(\bR_{BSk}\bR_{BSi} \right)\label{54}
		\end{align}
		because of the   uncorrelation between $ \bh_{BSk} $ and $ \bh_{BSi} $. 
		
		The third term is obtained as
		\begin{align}
			\EE\{|\bh_{RSk}^{\H}\bar{\bu}_{z}|^{2}\}=\tr(\bar{\bu}_{z}\bar{\bu}_{z}^{\H}\bR_{RSk})\label{55}.
		\end{align}

		Finally, we substitute \eqref{term2},  \eqref{est3},  \eqref{54},  \eqref{55}, and \eqref{normalization} into \eqref{SINR2}.

		\section{Proof of Proposition~\ref{LemmaGradients}}\label{lem2}
		First, we focus on the derivation of $\nabla_{\boldsymbol{\theta^{t}}}	\mathrm{SE} $, which is the complex gradient of the achievable sum SE with respect to $ \boldsymbol{\theta}^{t\ast}$.
		Based on \eqref{SINR4}, the derivation of \eqref{LowerBound} gives
		\begin{equation}
			\nabla_{\thetv^{t}}	\mathrm{SE}=\frac{1}{\tau_{c}}\sum_{k=1}^{K}\frac{I_{k}\nabla_{\boldsymbol{\theta}^{t}}S_{k}-S_{k}\nabla_{\boldsymbol{\theta}^{t}}I_{k}}{(1+\gamma_{k})I_{k}^{2}}.
		\end{equation}
		The computation of $\nabla_{\thetv^{t}}S_k$ for a given user $k$ takes into account that $\nabla_{\thetv^{t}}S_k = 0$ if $w_k=r$ as can be seen by \eqref{Sk} and the expression of $ \bR_{BSk} $. Hence, we aim to to find $\nabla_{\thetv^{t}}S_k$ only when $w_k=t$.
		
		\begin{align}
			d(S_{k})&=d\bigl(\tr(\bR_{BSk})^{2}\bigr)\nn\\
			&=2\tr(\bR_{BSk}d(\tr(\bR_{BSk})))\nn\\
			&=2\tr(\bR_{BSk})\tr(d(\bR_{BSk})).\label{eq:dSk}
		\end{align}
		
		The differential $d(\bR_{BSk})$ is derived as follows
		\begin{align}
			d(\bR_{BSk})
			& =\tilde{ \beta}_{BSk}\mathbf{R}_{\mathrm{BS}}\tr\bigl(\mathbf{A}_{t}\herm d(\bPhi_{t})+\mathbf{A}_{t}d\bigl(\bPhi_{t}\herm\bigr)\bigr).\label{dbsk}
		\end{align}
		By accounting that $\bPhi_{t}$ is diagonal, \eqref{dbsk} can be written as 
		\begin{align}
			d(\bR_{BSk})&=\tilde{ \beta}_{BSk}\mathbf{R}_{\mathrm{BS}}\bigl(\bigl(\diag\bigl(\mathbf{A}_{t}\herm\diag(\boldsymbol{{\beta}}^{t})\bigr)\bigr)\trans d(\boldsymbol{\theta}^{t})\nn\\
			&+\bigl(\diag\bigl(\mathbf{A}_{t}\diag(\boldsymbol{{\beta}}^{t})\bigr)\bigr)\trans d(\boldsymbol{\theta}^{t\ast})\bigr).\label{eq:dRk}
		\end{align}
		
		Substitution of \eqref{eq:dRk} into \eqref{eq:dSk} gives
		\begin{align}
			d(S_{k})	&=v_{k}\bigl(\bigl(\diag\bigl(\mathbf{A}_{t}\herm\diag(\boldsymbol{{\beta}}^{t})\bigr)\bigr)\trans d(\boldsymbol{\theta}^{t})\nn\\
			&+\bigl(\diag\bigl(\mathbf{A}_{t}\diag(\boldsymbol{{\beta}}^{t})\bigr)\bigr)\trans d(\boldsymbol{\theta}^{t\ast})\bigr),\label{eq:dSk:final}
		\end{align}
		where 
		\begin{align}
			v_{k}=2\tilde{ \beta}_{BSk}\tr(\bR_{BSk})\tr(\mathbf{R}_{\mathrm{BS}}).
		\end{align}
		Hence, from \eqref{eq:dSk:final}, we obtain that 
		\begin{align}
			\nabla_{\thetv^{t}}S_k&=\frac{\partial}{\partial{\thetv^{t\ast}}}S_{k}\nn\\
			&=v_{k}\diag\bigl(\mathbf{A}_{t}\diag(\boldsymbol{{\beta}}^{t})\bigr)
		\end{align}
		for $w_k=t$, which gives \eqref{derivtheta_t}. Similarly, we can prove \eqref{derivtheta_r}, but its proof is omitted for the sake of brevity.
		
		Regarding the differential $ 	d(I_{k}) $, we obtain from \eqref{Ik} that
		\begin{align}
			d(I_{k})&=2\tr(\bR_{BSk})d(\bR_{BSk})(\bR_{BSk})\tr(d(\bR_{BSk}))\nn\\
			&+	\!\!	\sum_{ i \in K_{t}/k} \!\!\tr\!\left(d(\bR_{BSk})\bR_{BSi} +\bR_{BSk}d(\bR_{BSi})\right)\nn\\
			&-\frac{K}{\bar{\lambda}^{2} \rho}(\sum_{i=1}^{K_{t}}\!\tr(d(\bR_{BSi})))(\tr(\bar{\bu}_{z}\bar{\bu}_{z}^{\H}\bR_{RSk})+\sigma_{c}^{2})\nn\\
			&+\frac{K}{\bar{\lambda} \rho ZL}\sum_{z=1}^{Z}\tr(\bar{\bu}_{z}\bar{\bu}_{z}^{\H}d(\bR_{RSk})), 
		\end{align}
		where the differentials equal to zero if $w_{k}\neq t$. Also, we have that $ d(\bR_{BSk}) $ has been obtained in \eqref{dbsk}, while $ d(\bR_{RSk}) $ is obtained similar to \eqref{eq:dRk} as
		\begin{align}
			d(\bR_{RSk})&=\tilde{ \beta}_{RSk} \bR_{\mathrm{R}}\bigl(\bigl(\diag\bigl(\mathbf{A}_{t}\herm\diag(\boldsymbol{{\beta}}^{t})\bigr)\bigr)\trans d(\boldsymbol{\theta}^{t})\nn\\
			&+\bigl(\diag\bigl(\mathbf{A}_{t}\diag(\boldsymbol{{\beta}}^{t})\bigr)\bigr)\trans d(\boldsymbol{\theta}^{t\ast})\bigr).
		\end{align}
		Thus, we can obtain $\nabla_{\thetv^{t}}I_k$ as
		\begin{align}
			\nabla_{\thetv^{t}}I_k=\frac{\partial}{\partial\boldsymbol{\theta}^{t\ast}}I_{k} & =\diag\bigl(\tilde{\mathbf{A}}_{kt}\diag(\boldsymbol{\mathbf{\beta}}^{t})\bigr),
		\end{align}
		where
		\begin{equation}
			\tilde{\mathbf{A}}_{kt}=\begin{cases}
				\bar{\nu}_{k}\mathbf{A}_{t}+\sum\nolimits _{i\in\mathcal{K}_{t}}^{K}\tilde{\nu}_{ki}\mathbf{A}_{t} & w_{k}=t\\
				\sum\nolimits _{i\in\mathcal{K}_{t}}\tilde{\nu}_{ki}\mathbf{A}_{t} & w_{k}\neq t,
			\end{cases}
		\end{equation}
		with $\bar{\nu}_{k}=2\tilde{ \beta}_{BSk}\tr(\bR_{BSk}\mathbf{R}_{\mathrm{BS}})-2\tilde{ \beta}_{BSk}\tr(\bR_{BSk})\tr(\mathbf{R}_{\mathrm{BS}})+\frac{K}{\bar{\lambda} \rho ZL}\tilde{ \beta}_{RSk}\sum_{z=1}^{Z}\tr(\bar{\bu}_{z}\bar{\bu}_{z}^{\H}\bR_{\mathrm{R}})$, 
		$\tilde{\nu}_{ki}=\tilde{ \beta}_{BSk}\tr\bigl(\mathbf{R}_{\mathrm{BS}}(\bR_{BSi}+\bR_{BSk})\bigr)-\frac{K}{\bar{\lambda}^{2} \rho}\tilde{ \beta}_{BSk}\tr(\mathbf{R}_{\mathrm{BS}})(\frac{1}{ZL}\sum_{z=1}^{Z}\tr(\bar{\bu}_{z}\bar{\bu}_{z}^{\H}\bR_{RSk})+\sigma_{c}^{2})$.
		
		The gradients  $\nabla_{\boldsymbol{\beta}^{u}}S_{k}$ and $\nabla_{\boldsymbol{\beta}^{u}}I_k$ with $ u=\{t,r\} $ are obtained similarly by noticing that
		\begin{subequations}\label{dRk_beta_t1}
			\begin{align}
				d(\mathbf{R}_{k}) & =\tilde{ \beta}_{BSk}\mathbf{R}_{\mathrm{BS}}\tr\bigl(\mathbf{A}_{t}\herm d(\bPhi_{t})+\mathbf{A}_{t}d\bigl(\bPhi_{t}\herm\bigr)\bigr)\\
				& =\tilde{ \beta}_{BSk}\mathbf{R}_{\mathrm{BS}}\bigl(\diag\bigl(\mathbf{A}_{t}\herm\diag(\btheta^{t})\bigr)^{\T}d(\boldsymbol{\beta}^{t})\nn\\
				&+\diag\bigl(\mathbf{A}_{t}\diag(\btheta^{t\ast})\bigr)^{\T}d(\boldsymbol{\beta}^{t})\bigr)\\
				& =2\tilde{ \beta}_{BSk}\mathbf{R}_{\mathrm{BS}}\Re\bigl\{\diag\bigl(\mathbf{A}_{t}\herm\diag(\btheta^{t})\bigr\}^{\T} d(\boldsymbol{\beta}^{t}).
			\end{align}
		\end{subequations}
		
		\section{Proof of Proposition~\ref{Proposition1}}\label{Prop1}
		By noticing that $ \bw_{z} $ is present only in the constraint \eqref{Maximization2313},
		the optimal solution is obtained by means of the Rayleigh quotient maximization \cite{Datta2010},
		as
		\begin{align}
			\bw_{z}^{\star}&=(\sigma^{2}_{r}\Id_{M}+\bA)^{-1} \al_{z}\ba(\bar{\theta}_{z})\ba^{\T}(\bar{\theta}_{z} ) \bar{\bu}_{z},
		\end{align}
		where we have defined $ \bA=\bR_{\mathrm{R}}\sum_{i=1}^{K}\tr(\bR_{BSi}\bR_{BS})(\tilde{ \beta}_{BR}+\tilde{ \beta}_{BSR}\tr(\bPhi_{w_{k}}^{\H}\bR_{\mathrm{RIS}}\bPhi_{w_{k}}\bR_{\mathrm{RIS}})) $.

		Having found the optimal $ \bw_{z} $, the problem $ 	(\mathcal{P}3) $ can be reformulated as
		\begin{subequations}
			\begin{align}
				(\mathcal{P}4)~~&\max_{\bu_{k}} 	\;		\sum_{z=1}^{Z}\tr(\bar{\bu}_{z}\bar{\bu}_{z}^{\H}\bR_{RSk})
				\label{Maximization4} \\
				&~~	\mathrm{s.t}~~\;\!	|\ba^{\T}(\bar{\theta}_{z} ) \bar{\bu}_{z}|^{2}\ge \bar{\gamma}^{r}	\label{Maximization24} \\
				&\;~~~~~~\!\sum_{z=1}^{Z}\|\bar{\bu}_{z}\|^{2}\le P_{\mathrm{max}},\label{Maximization214}	
			\end{align}
		\end{subequations}
		where the constraint \eqref{Maximization24} is obtained by rewriting \eqref{Maximization2313} for optimal $ \bw_{z} $ as
		\begin{align}
			&\!\!\!\!|\al_{z}|^{2}\bar{\bu}_{z}^{\H}\ba^{*}(\bar{\theta}_{z})\ba(\bar{\theta}_{z})^{\H} (\sigma^{2}_{r}\Id_{M}+\bA)^{-1} \ba(\bar{\theta}_{z})\ba^{\T}(\bar{\theta}_{z} ) \bar{\bu}_{z}\nn\\
			&\!\!\!\!= |\al_{z}|^{2}(\ba(\bar{\theta}_{z})^{\H} (\sigma^{2}_{r}\Id_{M}+\bA)^{-1} 	 \ba(\bar{\theta}_{z}))|\ba^{\T}(\bar{\theta}_{z} ) \bar{\bu}_{z}|^{2}\ge \gamma^{r}.
		\end{align}
		Meanwhile, we define
		\begin{align}
			\bar{\gamma}^{r}&=\frac{\gamma^{r}}{|\al_{z}|^{2}}(\ba(\bar{\theta}_{z})^{\H} (\sigma^{2}_{r}\Id_{M}+\bA)^{-1} \ba(\bar{\theta}_{z}))^{-1}\\
			&=\frac{\gamma^{r}}{|\al_{z}|^{2}}\Big(1-\frac{\ba(\bar{\theta}_{z})^{\H}\bA \ba(\bar{\theta}_{z})}{\sigma^{2}_{r}+\bA}\Big)^{-1},
		\end{align}
		where  the last equation is obtained with the help of the Sherman-Morrison formula, which gives \eqref{Maximization24}.
		
		The next step includes writing $ \bar{\bu}_{z} $ as a linear combination of $ \ba(\bar{\theta}_{z}) $, $ \bee_{z} $, and $ \br_{z} $, i.e., we have
		\begin{align}
			\bar{\bu}_{z}=\eta_{z,1}\ba^{*}(\bar{\theta}_{z})+\eta_{z,2}\bee_{z}+\br_{z},\label{uvector}
		\end{align}
		where $ \br_{z} $ is orthogonal to $ \ba^{*}(\bar{\theta}_{z})  $ and $ \bee_{z} $, i.e., we have $ \br_{z}^{\H}\ba^{*}(\bar{\theta}_{z})=0 $ and  $ \br_{z}^{\H}\bee_{z}=0 $. Now, problem $ (\mathcal{P}4) $ can be simplified to
		\begin{subequations}
			\begin{align}
				(\mathcal{P}5)~&\max_{\bar{\bu}_{z}} 	\;	\!\!	\sum_{z=1}^{Z}\tr(\eta_{z,1}^{2}\ba^{\T}(\bar{\theta}_{z})\bR_{RSk}\ba^{*}(\bar{\theta}_{z})+\eta_{z,2}^{2}\bee_{z}^{\H}\bR_{RSk}\bee_{z})
				\label{Maximization5} \\
				&~~	\mathrm{s.t}~~\;\!	|\eta_{z,1}|^{2}\ge \bar{\gamma}^{r}	\label{Maximization25} \\
				&\;~~~~~~\!\sum_{k=1}^{K}(|\eta_{z,1}|^{2}+|\eta_{z,2}|^{2})\le P_{\mathrm{max}}.\label{Maximization215}	
			\end{align}
		\end{subequations}
		After application of the method of Lagrange multiplier and KKT conditions, we result in the following optimal solution to problem $ (\mathcal{P}5) $
		\begin{align}
			\eta_{z,1}^{\star}&=\sqrt{\bar{\gamma}^{r}},\\
			\eta_{z,2}^{\star}&=\left\{
			\begin{array}{ll}
				-\frac{\ba^{\T}(\bar{\theta}_{z})\bR_{RSk}\ba^{*}(\bar{\theta}_{z})}{\bee_{z}^{\H}\bR_{RSk}\bee_{z}+\bar{\lambda}^{\star}}\sqrt{ \bar{\gamma}^{r}},&  \mathrm{if}~\bee_{z}^{\H}\bR_{RSk}\bee_{z}\ne 0,\\
				0, & \mathrm{if}~\bee_{z}^{\H}\bR_{RSk}\bee_{z}=0,
			\end{array} 
			\right. 
		\end{align}
		where $ \bar{\lambda}^{\star}\ge 0 $ is the  optimal Lagrange multiplier for constraint \eqref{Maximization215}, which satisfies the following condition
		\begin{align}
			\bar{\lambda}^{\star}(\sum_{z=1}^{Z}(|\eta^{\star}_{z,1}|^{2}+|\eta_{z,2}^{\star}|^{2})- P_{\mathrm{max}})=0.
		\end{align}
		Regarding $ \bee_{z} $, it is obtained via the Gram-Schmidt orthogonalization.
		
		\section{Proof of Proposition~\ref{Proposition2}}\label{Prop2}
		We  observe that the interference  from the radar to each UE, given by \eqref{Maximization4} can be written as
		\begin{align}
			&\sum_{z=1}^{Z}\tr(\bar{\bu}_{z}\bar{\bu}_{z}^{\H}\bR_{RSk})=\sum_{z=1}^{Z}\tr(|\eta_{z,1}\ba^{*}(\bar{\theta}_{z})+\eta_{z,2}\bee_{z}|^{2}\bR_{RSk})\\
			&=	\sum_{z=1}^{Z}\bar{\gamma}^{r}\ba^{T}(\bar{\theta}_{z})\bR_{RSk}\ba^{*}(\bar{\theta}_{z})\Big|1-\frac{\bee_{z}^{\H}\bR_{RSk}\bee_{z}}{\bee_{z}^{\H}\bR_{RSk}\bee_{z}+\bar{\lambda}^{\star}}\Big|^{2},
			\label{Maximization6}
		\end{align}
		where we have used \eqref{uvector}. In \eqref{Maximization6}, it is shown that the interference increases with $ \bar{\lambda}^{\star} $.
		Also, the complementary  slackness condition for $ \bar{\lambda}^{\star} $ can be written in another way as in \eqref{Maximization7}.
		\begin{figure*}
					\begin{align}
			\bar{\lambda}^{\star}\!\!\left(\sum_{z=1}^{Z}\!\bar{\gamma}^{r}\!\!\left(\!\!1\!+\!\frac{\ba^{\T}(\bar{\theta}_{z})\bR_{RSk}\ba^{*}(\bar{\theta}_{z})\bee_{z}^{\H}\bR_{RSk}\bee_{z}\!+\!\bar{\lambda}^{\star}}{\left(\bee_{z}^{\H}\bR_{RSk}\bee_{z}+\bar{\lambda}^{\star}\right)^{2}}\!\!\right)\!-\! P_{\mathrm{max}}	\!\!\right)\!	=\!0	\label{Maximization7}.
		\end{align}
		\hrulefill
	\end{figure*}

		According to \eqref{Maximization7}, we observe that an increase in $ P_{\mathrm{max}} $  requires a  decrease in $ \bar{\lambda}^{\star} $ to satisfy this slackness condition. Especially, when $ P_{\mathrm{max}} $ becomes very large, $ \bar{\lambda}^{\star} $ goes to zero, which means that the interference from the radar to each UE tends to zero with $ P_{\mathrm{max}} $. Hence, we obtain the relationship between the radar power budget and the communication SINR as given by Proposition \ref{Proposition2}.

	\end{appendices}
	
	\bibliographystyle{IEEEtran}
	
	\bibliography{IEEEabrv,mybib}

\end{document}